\documentclass[lettersize,journal]{IEEEtran}
\usepackage{amsmath,amsfonts}
\usepackage{mathtools} 
\usepackage{algorithmic}
\usepackage{algorithm}
\usepackage{array}
\usepackage[caption=false,font=normalsize,labelfont=sf,textfont=sf]{subfig}
\usepackage{textcomp}
\usepackage{stfloats}
\usepackage{url}
\usepackage{verbatim}
\usepackage{graphicx}
\usepackage{cite}
\usepackage{bm}
\usepackage{amsthm}
\usepackage{threeparttable}
\usepackage{multirow}
\usepackage{float}
\usepackage{color}
\newtheorem{lemma}{Lemma}
\newtheorem{theorem}{Theorem}
\begin{document}

\title{Uplink Signal Detection For Large-Scale MIMO-ISAC Systems}

\author{Jian Wang, Qiqiang Chen, Zheng Wang,~\IEEEmembership{Senior Member,~IEEE}, Fan Liu,~\IEEEmembership{Senior Member,~IEEE},\\Yili Xia,~\IEEEmembership{Member,~IEEE}, Yongming Huang,~\IEEEmembership{Fellow,~IEEE} and Chau Yuen,~\IEEEmembership{Fellow,~IEEE}
\thanks{Jian Wang, Qiqiang Chen, Zheng Wang, Liu Fan, Yili Xia, and Yongming Huang are with School of Information Science and Engineering, Southeast University, Nanjing 210096, China (e-mail: wjianzzu@gmail.com; q.chen@seu.edu.cn; wznuaa@gmail.com; f.liu@ieee.org; yilixia@seu.edu.cn; huangym@seu.edu.cn).

Chau Yuen is with the School of Electrical and Electronics Engineering, Nanyang Technological University, Singapore (e-mail: chau.yuen@ntu.edu.sg)
}
}

\maketitle

\begin{abstract}
Next-generation wireless communication systems are unifying large-scale multiple-input multiple-output (MIMO) and integrated sensing and communication (ISAC) to enhance sensing and communication performance.
In this paper, the signal detection problem for MIMO-ISAC systems is modeled as a mixed-integer least squares (MILS) problem.
To solve it efficiently, we propose a projection-based neighborhood search-aided alternating direction method of multipliers (P-NS-ADMM) detection scheme.
By theoretical analysis, we demonstrate that P-NS-ADMM achieves the same received diversity order as maximum likelihood (ML) detection.
For further complexity reduction, an iteration-based NS-ADMM (I-NS-ADMM) is proposed to remove the complex projection operation.
Complexity analysis shows its complexity advantage compared with P-NS-ADMM.
Moreover, to better estimate the sensing signals for I-NS-ADMM, a flexible mechanism of ADMM iterations is given.
Finally, simulations demonstrate the proposed NS-aided ADMM detection schemes have significant performance advantages in terms of both BER and NMSE.
\end{abstract}


\section{Introduction}

Recently, as a promising technology for next-generation wireless networks, large-scale multiple-input multiple-output (MIMO)-integrated sensing and communication (ISAC) has received extensive attentions from both industry and academia \cite{liu2022survey, 10012421, 9585321}.
On one hand, large-scale antenna arrays can provide high spatial resolution to enhance the sensing performance.
On the other hand, communication networks can utilize sensing techniques for beamforming design and channel estimation for the communication performance improvement \cite{liu2022integrated}.
Despite these advantages, current research predominantly focuses on transmitter-side signal processing, particularly beamforming \cite{liu2021cramer, liu2020joint, liu2022learning, wang2022exploiting} and waveform design \cite{liu2021dual, 9724187, 10771629}, with uplink signal detection remaining an important open question.

The receiver design in MIMO-ISAC systems is inherently more complex than that in radar-communication coexistence (RCC) architectures. 
In RCC systems, sensing and communication (S\&C) functions operate on separate platforms, which simplifies the design challenge to either suppress sensing interference for communication or mitigate communication interference for sensing operations \cite{zheng2017adaptive}.
In contrast, MIMO-ISAC receivers need to perform uplink communication signal detection and target parameter estimation simultaneously.
This requirement leads to significant mutual interference between S\&C signals, thus resulting in substantial challenges in receiver design.

Recently, several contributions have been made to the uplink MIMO-ISAC systems \cite{ouyang2022performance, liu2020road, temiz2021dual, dong2023joint, yu2024addressing}.
Specifically, the authors of \cite{ouyang2022performance} demonstrated that MIMO-ISAC receivers provide superior degrees of freedom for both S\&C compared to conventional frequency-division systems.
\cite{liu2020road} assumed that the sensing signals are received prior to the uplink communication signals.
This enables the prior estimation of sensing signal parameters, thereby facilitating the problem formulation.
Meanwhile, several successive interference cancellation (SIC)-based receiver designs have been proposed to mitigate mutual interference between S\&C signals in \cite{ouyang2022performance}, \cite{temiz2021dual}.
However, these SIC-based schemes not only suffer from the error propagation but also impede the derivation of the residual signal's distribution, hindering the optimal sensing estimation design.
The work in \cite{dong2023joint} proved that the SIC-based scheme is sub-optimal and formulated a tailored minimum mean squared error (MMSE) estimator for target estimation.
However, it only concerns single-user and single-target scenarios.
In \cite{yu2024addressing}, the authors proposed a projection-based scheme and applied semidefinite relaxation (SDR) for communication signal detection.
Unfortunately, the severe performance degradation of SDR in higher-order modulations (e.g., $\geq$ 16-QAM) substantially restricts its practical deployment.

Recent research has demonstrated the effectiveness of the alternating direction method of multipliers (ADMM) in solving both convex and non-convex optimization problems, owing to its simplicity and guaranteed convergence \cite{boyd2011distributed}.
The application of ADMM to MIMO detection was first proposed by Takapoui et al. \cite{takapoui2020simple}, leading to subsequent developments including ADMIN \cite{shahabuddin2021admm} and PS-ADMM \cite{zhang2022efficient} for large-scale MIMO systems.
However, these ADMM-based approaches inherently suffer from performance degradation as they have to relax the original discrete integer problem into a continuous domain.
On the other hand, the neighborhood search (NS) algorithms have emerged as an effective way to provide significant gains over the linear MMSE detection with low complexity cost \cite{sah2017sequential}.
Among these, likelihood ascent search (LAS) \cite{vardhan2008low} is the first reported algorithm that starts with an initial estimate and then searches for the optimal solution in the neighboring set by minimizing the maximum likelihood (ML) cost.
Based on it, techniques such as  multistage LAS (MLAS) \cite{mohammed2008low}, reactive tabu search (RTS) \cite{srinidhi2009low}, and unconstrained LAS (ULAS) \cite{sah2017unconstrained} have been given for better BER performance.
Despite these advances, theoretical analysis of the NS-based algorithms is essentially limited due to their uncertain number of iterations.
Recent work \cite{chen2025decentralized} addressed this limitation through the decentralized LAS (DLAS) mechanism, which applies NS to iterative algorithms such as ADMM and provably achieves full received diversity order.

In this paper, by exploiting ideas of the DLAS mechanism, we propose two NS-aided ADMM signal detectors for large-scale MIMO-ISAC systems.
Both of them demonstrate favorable BER and NMSE performance.
In summary, the main contributions of this paper are as follows:
\begin{itemize}
\item First of all, to address the mixed-integer least squares (MILS) problem in uplink MIMO-ISAC systems, we propose a projection-based NS-aided ADMM (P-NS-ADMM) scheme. 
It analytically eliminates continuous sensing variables via orthogonal projection, transforming the complex MILS problem into a lower-dimensional integer least squares (ILS) problem for efficient solving.
\item Secondly, we provide a rigorous theoretical analysis for the proposed P-NS-ADMM. 
We prove that it achieves the same received diversity order as the optimal ML detection, guaranteeing its reliability in high SNR regimes.
\item Thirdly, to reduce the high computational cost of projection, we propose a low-complexity iteration-based NS-aided ADMM (I-NS-ADMM) scheme that solves the MILS problem directly.
A flexible ADMM iteration mechanism is introduced to refine the update of continuous sensing variables, ensuring superior estimation accuracy.
\end{itemize}

The rest of this paper is organized as follows.
Section II introduces the sensing and communication signal model and presents the basic framework of ADMM for MIMO detection.
In Section III, the detection problem for the MIMO-ISAC system is formulated.
In Section IV, the P-NS-ADMM scheme is proposed.
Then, its received diversity order is derived in section V.
In Section VI, a low-complexity scheme, I-NS-ADMM, is proposed, along with a flexible mechanism for ADMM iterations that enhances sensing estimation.
Section VII shows simulations of the proposed NS-aided detection for uplink large-scale MIMO-ISAC systems. 
Finally, Section VIII concludes the paper.

\emph{Notation:} Matrices and column vectors are denoted by upper and lowercase boldface letters, and the transpose, inverse of a matrix $\mathbf{B}$ by $\mathbf{B}^T$ and $\mathbf{B}^{-1}$, respectively.
We use $\mathbf{b}_i$ for the $i$-th column of the matrix $\mathbf{B}$, $b_{i, j}$ for the entry in the $i$-th row and $j$-th column of the matrix $\mathbf{B}$.
$\odot$ represents the Hadamard product, which performs element-wise multiplication between two vectors.
Additionally, $\mathrm{Tr}(\cdot)$, $E[\cdot]$, and $Var[\cdot]$ denote the trace of the matrix, expectation, and variance.
$\text{diag}(\mathbf{B})$ extracts the diagonal elements of the square matrix $\mathbf{B}$, and $\lceil \rfloor$ rounds to the closest integer.
$\mathbf{B} \succeq 0$ denotes the matrix $\mathbf{B}$ is positive semi-definite.
$\Re{(\cdot)}$ and $\Im{(\cdot)}$ indicate the real and imaginary components, respectively.
Finally, the superscript $(\cdot)^k$ denotes the iteration index.

\section{Preliminary}

This section presents the sensing and communication signal models in the MIMO-ISAC system, along with the basic ADMM framework for subsequent algorithm design.

\subsection{System Model}
\label{sec2}

As depicted in Fig.~\ref{scenario}, we consider a MIMO-ISAC system, which consists of $M_t$ single antenna sensing targets, $U$ single antenna communication user equipment (UE), and a MIMO-ISAC base station (BS) equipped with $N_t$ transmit antennas and $N_r$ receive antennas.
Assuming perfect self-interference cancellation in full-duplex operation, we focus on handling the mutual interference.

To begin with, we denote the $M_t$ ISAC downlink streams at time $t$ as $\mathbf{s}_{0}(t) = [s_1(t), \ldots, s_{M_t}(t)]^T \in \mathbb{C}^{M_t}$.
\begin{figure}[t]
	\centering
	\includegraphics[width=0.44\textwidth]{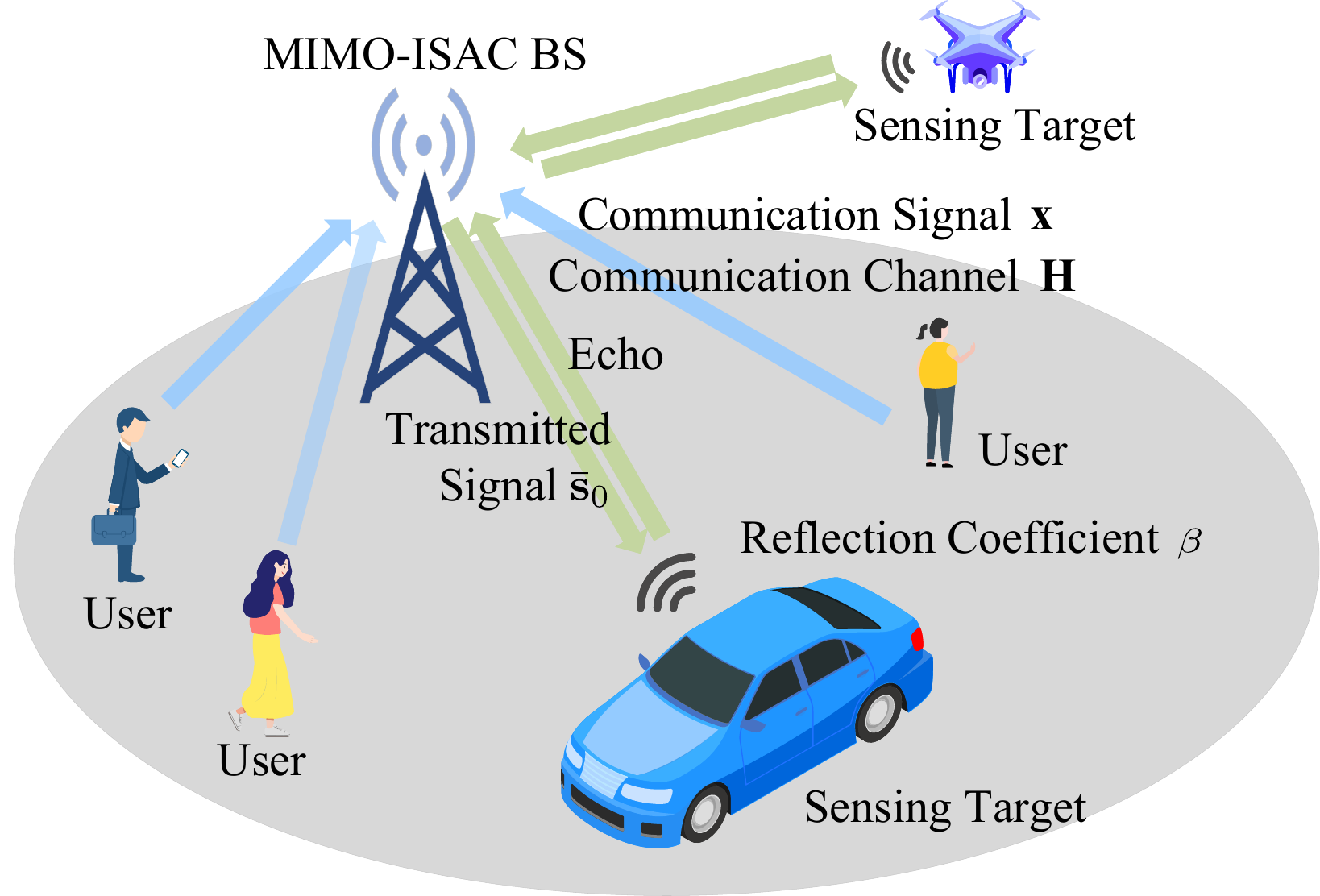}
	\caption{The uplink signal detection in MIMO-ISAC systems, where BS receives communication and sensing signals simultaneously.}
	\vspace{-0.4cm}
	\label{scenario}
\end{figure}
The transmitted signal can be expressed as
\begin{equation}
	\bar{\mathbf{s}}_0(t) = \mathbf{Fs}_{0}(t) \in \mathbb{C}^{N_t}
\end{equation}
with $\mathbf{F} \in \mathbb{C}^{N_t \times M_t}$ denoting the transmit beamforming matrix.
Accordingly, the reflected echoes received at the BS is given by
\begin{equation}
	\bar{\mathbf{y}}_s(t) = \sum_{m = 1}^{M_t} \beta_m \mathbf{a}_r(\theta_m) \mathbf{a}_t^H(\theta_m) \bar{\mathbf{s}}_0(t) + \bar{\mathbf{n}}_s(t),
	\label{rsensing}
\end{equation}
where $\beta_m$ is the reflection coefficient of the $m$-th target, and $\theta_m$ denotes the angle between the $m$-th target and the BS.
$ \bar{\mathbf{n}}_s$ represents the additive white Gaussian noise (AWGN) with zero mean and covariance matrix $\sigma^2_{s} \mathbf{I}_{N_r}$.
$\mathbf{a}_t(\theta_m)$ and $\mathbf{a}_r(\theta_m)$ are transmit and receive steering vectors, respectively, of the BS antenna array, taking the form
\begin{equation}
	\mathbf{a}_t(\theta) = \sqrt{\frac{1}{N_t}} \left[ 1, e^{-j \pi cos \theta}, \ldots, e^{-j \pi (N_t-1)cos \theta } \right]^T,
	\label{steert}
\end{equation}
\begin{equation}
	\mathbf{a}_r(\theta) = \sqrt{\frac{1}{N_r}} \left[ 1, e^{-j \pi cos \theta}, \ldots, e^{-j \pi (N_r-1)cos \theta } \right]^T,
	\label{steerr}
\end{equation}
where the half-wavelength antenna spacing for the uniform linear array (ULA) is employed.
Based on the prediction of the angles, the $m$-th column of $\mathbf{F}$ is given as
\begin{equation}
	\mathbf{f}_m = \mathbf{a}_t \left( \hat{\theta}_m\right), \ \forall m.
	\label{beamforming}
\end{equation}
Given the fact that we focus our attention on the estimation of reflection coefficient $\beta_m$, we assume the angle of the target is perfectly tracked \cite{dong2023joint}, so as to $\mathbf{f}_m = \mathbf{a}_t (\theta _m)$.

Exploiting the asymptotic orthogonality of ULA steering vectors in the massive MIMO systems ($N_r \to \infty$) \cite{ngo2015massive}, i.e., $|\mathbf{a}_r^H(\theta)\mathbf{a}_r(\phi)| \rightarrow 0$ for $\theta \neq \phi$, the inter-beam interference can be neglected.
Consequently, the reflected signal for $M_t$ targets can be rewritten as
\begin{align}
	\bar{\mathbf{y}}_s(t) &= \sum_{m = 1}^{M_t} \beta_m \mathbf{a}_r(\theta_m) \mathbf{a}_t^H(\theta_m) \mathbf{f}_m s_m(t) + \bar{\mathbf{n}}_s(t) \nonumber \\
	&=\bar{\mathbf{A}}\bar{\mathbf{s}}(t) + \bar{\mathbf{n}}_s(t),
	\label{rsensing2}
\end{align}
where ${\bar{\bf A}} = \left[ {{{\bf{a}}_r}({\theta _1}), \ldots ,{{\bf{a}}_r}({\theta _{M_t}})} \right] \in {\mathbb{C}^{{N_r} \times M_t}}$ denotes a matrix consisting of the receive steering vectors and  $\bar{\mathbf{s}}(t) = [\beta_1 s_1(t), \ldots, \beta_{M_t} s_{M_t}(t)]^T \in \mathbb{C}^{M_t}$ consists of downlink streams and reflection coefficient.
In subsequent analysis, we refer to $\bar{\mathbf{s}}(t)$ as the sensing signal.

For communication signals, let ${\bar{\bf x}}(t) \in {{\mathcal{O}}^U}$ denote the complex-valued uplink communication signals.
Then, the corresponding received communication signal vector ${{\bar{\bf y}}_c}(t) \in {\mathbb{C}^{{N_r}}}$ at BS is given by
\begin{equation}
	{{\bar{\bf y}}_c}(t) = {\bar{\bf H}}{\bar{\bf x}}(t) + {\bar{\bf n}_{c}}(t).
	\label{comproblem}
\end{equation}
Here, $\bar{\mathbf{x}}(t) \in \mathcal{O}^{U}$ represents the transmitted vector from the discrete complex $L$-quadrature amplitude modulation (QAM) constellation set $\mathcal{O}^{U}$, ${\bar{\bf H}} \in {\mathbb{C}^{{N_r} \times U}}$ is the Rayleigh fading channel matrix whose entries follow $\mathcal{CN}(0, 1)$, and ${\bar{\bf n}_{c}}(t) \in {\mathbb{C}^{{N_r}}}$ denotes AWGN with zero mean and covariance matrix $\sigma^2_{c} \mathbf{I}_{N_r}$.

\subsection{ADMM-based Algorithm Framework}
ADMM is a widely used numerical method for solving convex optimization problems by decomposing them into smaller sub-problems \cite{boyd2011distributed}.
Consider the standard optimization problem formulation:
\begin{equation}
\min_{\mathbf{x}, \mathbf{z}} f(\mathbf{x}) + g(\mathbf{z}) \quad \text{s.t. } \mathbf{x} = \mathbf{z},
\label{admmproblem}
\end{equation}
where $f(\mathbf{x})$ is the primary objective function, and $g(\mathbf{z})$ acts as the indicator function for a convex set $\mathcal{C}$:
\begin{equation}
g(\mathbf{z}) = \begin{cases} 0, & \text{if } \mathbf{z} \in \mathcal{C}, \\ \infty, & \text{otherwise.} \end{cases}
\end{equation}
Through the introduction of  Lagrange multiplier $\bm{\lambda}$ and penalty parameter $\rho > 0$, the constraint is incorporated into the objective function, leading to the following scaled augmented Lagrangian function
\begin{equation}
	\mathcal{L}(\mathbf{x}, \mathbf{z}, \bm{\lambda}) = f(\mathbf{x}) + g(\mathbf{z}) + \frac{\rho}{2}\|\mathbf{x}-\mathbf{z} + \bm{\lambda}\|^2.
	\label{augla}
\end{equation}
The ADMM algorithm minimizes \eqref{augla} by updating the variables sequentially:
\begin{subequations}
	\label{equopt}
	\begin{equation}
	\mathbf{x}^{k} = \arg \min_{\mathbf{x}} f(\mathbf{x}) + \frac{\rho}{2}\|\mathbf{x} - \mathbf{z}^{k-1} + \boldsymbol{\lambda}^{k-1}\|^2,
		\label{admmx}
	\end{equation}
	\begin{equation}
\mathbf{z}^{k} = \arg \min_{\mathbf{z}} g(\mathbf{z}) + \frac{\rho}{2}\|\mathbf{x}^{k} - \mathbf{z} + \boldsymbol{\lambda}^{k-1}\|^2,
		\label{admmzz}
	\end{equation}
	\begin{equation}
\boldsymbol{\lambda}^{k} = \boldsymbol{\lambda}^{k-1} + \mathbf{x}^{k} - \mathbf{z}^{k},
		\label{admmlambda}
	\end{equation}
\end{subequations}
For the indicator function $g(\mathbf{z})$, the $\mathbf{z}$-update reduces to the Euclidean projection onto $\mathcal{C}$, i.e., $\mathbf{z}^{k} = \Pi_{\mathcal{C}}(\mathbf{x}^{k} + \boldsymbol{\lambda}^{k-1})$.
Iterating these steps until convergence yields the solution to \eqref{admmproblem}.

\section{Problem Formulation}

In this section, we formulate the detection problem for the MIMO-ISAC systems.
Specifically, the received signal $\bar{\mathbf{y}}(t)$ is the superposition of the sensing component $\bar{\mathbf{y}}_s(t)$ and the communication component $\bar{\mathbf{y}}_c(t)$, expressed as
\begin{equation}
	\bar{\mathbf{y}}(t) = {\bar{\bf y}}_s(t) + {\bar{\bf y}}_c(t) = \bar{\mathbf{A}}\bar{\mathbf{s}}(t) + {\bar{\bf H}}{\bar{\bf x}}(t) +  \bar{\mathbf{n}}(t),
	\label{cy}
\end{equation}
where $\bar{\mathbf{n}}(t) = \bar{\mathbf{n}}_s(t) + \bar{\mathbf{n}}_c(t) \in {\mathbb{C}^{{N_r}}}$  denotes AWGN with zero mean and covariance matrix $\sigma^2_{\bar{n}} \mathbf{I}_{N_r}$.
Since our analysis focuses on a specific time instant, we simplify the notation by omitting the time index $t$ in subsequent derivations.
Based on this model, the receiver is tasked with jointly recovering the discrete communication vector $\bar{\mathbf x}$ and estimating the continuous sensing vector $\bar{\mathbf s}$ to determine the reflection coefficient $\beta_m$, assuming perfect knowledge of the communication channel $\bar{\mathbf H}$ and the sensing steering matrix $\bar{\mathbf{A}}$.
Additionally, we assume that $\bar{\mathbf{x}}$ and $\bar{\mathbf{s}}$ are statistically independent, where $\bar{\mathbf{x}}$ is uniformly distributed over the constellation set and $\bar{\mathbf{s}}$ follows a uniform prior.

Theoretically, the optimal decision criterion for minimizing error probability under Bayesian inference is the MAP criterion \cite{albreem2019massive}.
By applying Bayes’ theorem, we obtain:
\begin{equation}
	P(\bar{\mathbf{x}}, \bar{\mathbf{s}} | \bar{\mathbf{y}}) = \frac{P(\bar{\mathbf{y}} | \bar{\mathbf{x}}, \bar{\mathbf{s}}) P(\bar{\mathbf{x}})P(\bar{\mathbf{s}}) }{P({\bar{\mathbf{y}}})}.
\end{equation}
Given the uniform assumptions stated above, the prior probabilities $P(\bar{\mathbf{x}})$ and $P(\bar{\mathbf{s}})$ are constants.
Therefore, the MAP criterion reduces to maximizing the likelihood function $P(\bar{\mathbf{y}} | \bar{\mathbf{x}}, \bar{\mathbf{s}})$.
Given the additive white Gaussian noise $\bar{\mathbf{n}} \sim \mathcal{CN}(\mathbf{0}, \sigma_{\bar{n}}^2\mathbf{I}_{N_r})$, the conditional probability density function (PDF) becomes
\begin{equation}
	P(\bar{\mathbf{y}} | \bar{\mathbf{x}}, \bar{\mathbf{s}}) = \frac{1}{(\pi \sigma_{\bar{n}}^2)^{N_r}} e^{-||\bar{\mathbf{y}} - \bar{\mathbf{A}}\bar{\mathbf{s}} - \bar{\mathbf{H}}\bar{\mathbf{x}}||^2}.
\end{equation}
Consequently, maximizing $P(\bar{\mathbf{y}} | \bar{\mathbf{x}}, \bar{\mathbf{s}})$ is equivalent to minimizing $\|\bar{\mathbf{y}} - \bar{\mathbf{A}}\bar{\mathbf{s}} - \bar{\mathbf{H}}\bar{\mathbf{x}}\|^2$, which leads to the following signal detection problem in MIMO-ISAC systems
\begin{equation}
	(\hat{\mathbf{x}}, \hat{\mathbf{s}}) = \underset {\bar{\mathbf{x}} \in \mathcal{O}^U, \bar{\mathbf{s}} \in \mathbb{C}^{M_t}} {\text{arg  
			min}} \left\| \bar{\mathbf{y}} - \bar{\mathbf{A}}\bar{\mathbf{s}} - \bar{\mathbf{H}}\bar{\mathbf{x}}\right\|^2.
	\label{CMILS}
\end{equation}
Here, to facilitate the subsequent analysis, we express (\ref{cy}) as its real-valued counterpart
\begin{equation}
	\mathbf{y} = \mathbf{As} + \mathbf{Hx} + \mathbf{n},
	\label{original}
\end{equation}
where
\begin{align}
	\mathbf{H}=\left[\hspace{-.3em}
	\begin{array}{cc}
		\Re(\bar{\mathbf{H}}) & -\Im(\bar{\mathbf{H}}) \\ \Im(\bar{\mathbf{H}}) & \Re(\bar{\mathbf{H}})
	\end{array}
	\right]\hspace{-.3em},
	\mathbf{A}=
	\left[\hspace{-.3em}
	\begin{array}{cc}
		\Re(\bar{\mathbf{A}}) & -\Im(\bar{\mathbf{A}}) \\ \Im(\bar{\mathbf{A}}) & \Re(\bar{\mathbf{A}})
	\end{array}
	\right]\hspace{-.3em},
	\label{ctor}
\end{align}
$\mathbf{y}=\left[\Re(\bar{\mathbf{y}}) ; \Im(\bar{\mathbf{y}})\right]$,
$\mathbf{x}=\left[\Re(\bar{\mathbf{x}}) ; \Im(\bar{\mathbf{x}})\right]$,
$\mathbf{s}=\left[\Re(\bar{\mathbf{s}}) ; \Im(\bar{\mathbf{s}})\right]$, and $\mathbf{n}=\left[\Re(\bar{\mathbf{n}}) ; \Im(\bar{\mathbf{n}})\right]$.
Then, the problem of MIMO-ISAC detection in (\ref{CMILS}) becomes
\begin{equation}
	(\hat{\mathbf{x}}, \hat{\mathbf{s}}) = \underset {\mathbf{x} \in \mathcal{X}^{K}, \mathbf{s} \in \mathbb{R}^{M}} {\text{arg  
			min}} \left\| \mathbf{y} - \mathbf{As} - \mathbf{Hx} \right\|^2.
	\label{MILS}
\end{equation}
For simplicity of notation, from this point onward, let $K=2U$, $N=2N_r$ and $M=2M_t$.
In this way, the complex constellation $\mathcal{O}^U$ is transformed into a real-valued $\sqrt{L}$-amplitude-shift keying (ASK) constellation set $\mathcal{X}^{K}$, defined as $\mathcal{X} = \{ \pm1, \pm3, \ldots, \pm(\sqrt{L} - 1)\}$.

Since $\mathbf{x}$ is constellation points constrained by a discrete state space, and $\mathbf{s}$ resides in continuous real-valued space, the detection problem in (\ref{MILS}) is essentially a mixed-integer least squares (MILS) problem \cite{chang2007miles}, which belongs to a classic non-convex optimization problem.

\section{The Proposed Projection-based NS-aided ADMM Dection Scheme}

To solve the problem in \eqref{MILS}, inspired by \cite{yu2024addressing}, \cite{wang2021new}, the problem can be equivalently translated into the following form:
\begin{equation}
	(\hat{\mathbf{x}}, \hat{\mathbf{s}}) =\! \underset {\mathbf{x} \in \mathcal{X}^{K}, \mathbf{s} \in \mathbb{R}^{M}} {\text{arg  
			min}}\!\! (\mathbf{s} - \hat{\mathbf{s}}(\mathbf{x}))^T \mathbf{\Xi}^{-1} (\mathbf{s} -\hat{\mathbf{s}}(\mathbf{x})) + \left\| \mathbf{P}(\mathbf{y} - \mathbf{Hx}) \right\|^2 \!,
	\label{pro}
\end{equation}
where $\mathbf{\Xi} = (\mathbf{A}^T\mathbf{A})^{-1}$, $ \mathbf{P} = \mathbf{I}_N - \mathbf{A}(\mathbf{A}^T\mathbf{A})^{-1}\mathbf{A}^T$, and $\hat{\mathbf{s}}$ is a function of $\mathbf{x}$.
Here, $\mathbf{P}$ functions as an orthogonal projection operator onto the null space of $\mathbf{A}$. 
By satisfying the property $\mathbf{P}\mathbf{A} = \mathbf{0}$, this projection removes the sensing interference component $\mathbf{A}\mathbf{s}$.
Since $\mathbf{\Xi } \succeq 0$, the first term in \eqref{pro} is nonnegative with minimum zero.
Meanwhile, given that the second term is independent of $\mathbf{s}$, the problem in \eqref{pro} is equivalent to minimizing the second term, which can be rewritten as
\begin{equation}
	\hat{\mathbf{x}} = \underset {\mathbf{x} \in \mathcal{X}^{K}} {\text{ min}} \| \tilde{\mathbf{y}} - \tilde{\mathbf{H}}\mathbf{x} \|^2
	\label{comdetection}
\end{equation}
with $\tilde{\mathbf{H}} = \mathbf{PH}$ and $\tilde{\mathbf{y}} = \mathbf{Py}$.
Consequently, one solution for solving \eqref{comdetection} is to firstly obtain $\hat{\mathbf{x}}$, and then substitute $\hat{\mathbf{x}}$ into \eqref{MILS} to get $\hat{\mathbf{s}}$.

However, despite the dimensionality reduction via projection, the problem in \eqref{comdetection} remains non-convex due to the persistent discrete constraints on $\mathcal{X}^K$.
Theoretically, the optimum detector for \eqref{comdetection} is the ML detector, which can achieve full received diversity order \cite{4217646}.
The received diversity order serves as a reliability metric that characterizes the decaying rate of transmission error probability with increasing SNR, which can be expressed as
\begin{equation}
	d = \lim_{\sigma_c^2 \to 0}\frac{\mathrm{log}\left(\bar{P}_{e}\right)}{\mathrm{log}(\sigma_c^2)}=\lim_{\rho \to \infty}-\frac{\mathrm{log}\left(\bar{P}_{e}\right)}{\mathrm{log} \rho},
\end{equation}
where $\bar{P}_{e}$ denotes the pair-wise error probability (PEP) in transmission and $\rho=1/\sigma_c^2$ denotes the signal-to-noise ratio (SNR).

\subsection{Non-convex Relaxation for ADMM}

To solve the non-convex problem in \eqref{comdetection}, we first relax the constraints $\mathcal{X}^{K} = \{ \pm1, \pm3, \ldots, \pm(\sqrt{L} - 1)\}^{K}$ to box constrains $[- \sqrt{L} + 1, \sqrt{L} - 1]^{K}$.
The problem in (\ref{comdetection}) becomes
\begin{align}
	\underset {\mathbf{x} \in \mathbb{R}^{K}, \mathbf{z} \in \mathcal{R}^{K}_{\mathcal{X}}} {\text{ min}} \quad & \| \tilde{\mathbf{y}} - \tilde{\mathbf{H}}\mathbf{x} \|^2, \nonumber \\
	\text{s.t.} \quad & \mathbf{x} = \mathbf{z},
	\label{relax}
\end{align}
where $\mathcal{R}_{\mathcal{X}} \in [- \sqrt{L} + 1, \sqrt{L} - 1]$ is the continuous interval.
Subsequently, the scaled augmented \emph{Lagrangian} of (\ref{relax}) can be expressed as
\begin{equation}
	\mathcal{L} (  \mathbf{x}, \mathbf{z}, \bm{\lambda}  ) = \frac{1}{2} \|  \tilde{\mathbf{y}} - \tilde{\mathbf{H}}\mathbf{x}  \|^2 + \frac{\rho}{2} \|  \mathbf{x} - \mathbf{z} + \bm{\lambda}\|^2,
	\label{lagrangian}
\end{equation}
where ${\bm{\lambda }} \in {\mathbb{R}^{K}}$ and $\rho > 0$ are the Lagrangian multiplier and penalty parameter, respectively.
Similar to \eqref{equopt}, the update of the parameter are as follows
\begin{subequations}
\begin{equation}
	\label{xopt}
	{\mathbf{x}}^{k}  = \underset {\mathbf{x} \in \mathbb{R}^{K}} {\text{arg  min}}\; \frac{1}{2} \|  \tilde{\mathbf{y}} - \tilde{\mathbf{H}}\mathbf{x}  \|^2 + \frac{\rho}{2} \|  \mathbf{x} - \mathbf{z}^{k-1} + \bm{\lambda}^{k-1} \|^2,
\end{equation}
\begin{equation}
	\label{zopt}
	{\mathbf{z}}^{k}  = \underset {{\mathbf{z}} \in \mathcal{R}^{K}_{\mathcal{X}}} {\text{arg  min}} \; \frac{\rho}{2} \|  \mathbf{x}^{k} - \mathbf{z} + \bm{\lambda}^{k-1} \|^2,
\end{equation}
\begin{equation}
	\label{lamopt}
	\bm{\lambda}^{k} = \bm{\lambda}^{k-1} + \mathbf{x}^{k} - \mathbf{z}^{k}.
\end{equation}
\end{subequations}
Specifically, since the objective function in (\ref{xopt}) is a convex quadratic function with respect to ${\mathbf{x}}$, we can set the gradient of the corresponding augmented Lagrangian function about ${\mathbf{x}}$ to be zero.
Therefore, we have
\begin{equation}
	\label{xupdate}
	\mathbf{x}^{k} = (\tilde{\mathbf{H}}^T\tilde{\mathbf{H}} + \rho\mathbf{I})^{-1} \left( \tilde{\mathbf{H}}^T\tilde{\mathbf{y}} + \rho(\mathbf{z}^{k-1} - \bm{\lambda}^{k-1})\right).
\end{equation}
Moreover, since the objective function in (\ref{zopt}) is convex regarding to $\mathbf{z}$, by setting the gradient about ${\mathbf{z}}$ to be zero and projecting each entry of ${\mathbf{z}}$ onto $\mathcal{R}_{\mathcal{X}}$, we can obtain
\begin{equation}
	\label{zupdate}
	\mathbf{z}^{k} = \bm{\Pi}_{\mathcal{R}_{\mathcal{X}}^{K}} (\mathbf{x}^{k} + \bm{\lambda}^{k-1}),
\end{equation}
where the projection is given by
\begin{equation}
	\bm{\Pi}_{\mathcal{R}_{\mathcal{X}}^{K}}(a) = 
	\begin{cases} 
		\sqrt{L} - 1, & \text{if } a> \sqrt{L} - 1, \\
		-\sqrt{L} + 1, & \text{if } a < -\sqrt{L} + 1, \\
		a, & \text{otherwise.}
	\end{cases}
	\label{projection}
\end{equation}
By doing this, we relax the nonconvex set to a convex set and give an expression for the parameter update in the ADMM framework.

\subsection{NS-aided ADMM}

To further improve the performance, we then introduce the NS mechanism.
Considering the complexity, the NS technique searches only a one-symbol neighborhood, where vectors differ from the initial solution in exactly one entry.
Typically, the objective of NS is to minimize the following ML cost function
\begin{equation}
	\mathcal{F}_{\text{ML}} = \|  \tilde{\mathbf{y}} - \tilde{ \mathbf{H}}\mathbf{x} \|^2,
	\label{costfunction}
\end{equation}
which serves a metric to evaluate the estimated vector $\mathbf{x}$.
Given that the ML detection achieves the minimum $\mathcal{F}_{\text{ML}}$, a lower ML cost function value corresponds to a more accurate estimated vector $\mathbf{x}$.

Here, the final output of ADMM is used as an initial solution for the NS mechanism, namely, $\mathbf{x}^0 = \lceil \mathbf{x}^{T_{\text{max}}} \rfloor_{\mathcal{Q}} \in \mathcal{X}^{K}$.
We then update $\mathbf{x}^{0}$ follows the rule
\begin{equation}
	\mathbf{x}^{\text{new}} = \mathbf{x}^{0} + d_i \mathbf{u}_i,
\end{equation}
where $\mathbf{u}_i \in \mathbb{R}^{K}$ denotes the unit vector with its $i$-th element set to one and other elements zero.
To ensure both $\mathbf{x}^{\text{new}}$ and $\mathbf{x}^{0}$ remain within the constellation space $\mathcal{X}^{K}$, the value of $d_i$ must be constrained to integers that are multiples of 2.

To determine the best value for $d_i$, which minimizes the ML cost in (\ref{costfunction}), the following ML cost difference is considered
\begin{align}
	&\Delta \mathcal{F}_{\text{ML}}(d_i) = \| \tilde{\mathbf{y}} - \tilde{\mathbf{H}}\mathbf{x}^{\text{new}} \|^2 - \| \tilde{\mathbf{y}} - \tilde{\mathbf{H}}\mathbf{x}^0 \|^2 \nonumber \\
	&= d_i^2 \mathbf{u}_i^T \tilde{\mathbf{H}}^T\tilde{\mathbf{H}}\mathbf{u}_i \!+\! d_i \mathbf{u}_i^T\tilde{\mathbf{H}}^T\tilde{\mathbf{H}} \mathbf{x}^0 \!+\! d_i \mathbf{x}^0 \tilde{\mathbf{H}}^T\tilde{\mathbf{H}}\mathbf{u}_i \!-\! 2d_i \tilde{\mathbf{y}}^T\tilde{\mathbf{H}}\mathbf{u}_i
	\nonumber \\
	&= d_i^2(\tilde{\mathbf{H}}^T\tilde{\mathbf{H}})_{i,i} + 2d_i (\tilde{\mathbf{H}}^T\tilde{\mathbf{H}}\mathbf{x}^0)_i - 2d_i(\tilde{\mathbf{H}}^T\tilde{\mathbf{y}})_i \nonumber \\
	&= d_i^2 \tilde{g}_{i, i} - 2d_i (\tilde{\mathbf{H}}^T\tilde{\mathbf{y}} - \tilde{\mathbf{H}}^T\tilde{\mathbf{H}}\mathbf{x}^0)_i \nonumber \\
	&= d_i^2 \tilde{g}_{i, i} - 2d_i \tilde{e}_i,
	\label{deltaF}
\end{align} 
where $\tilde{\mathbf{G}} = \tilde{\mathbf{H}}^T\tilde{\mathbf{H}} \in \mathbb{R}^{K \times K}$ denotes the Gram matrix and $\tilde{\mathbf{e}} = \tilde{\mathbf{H}}^T(\tilde{\mathbf{y}} - \tilde{\mathbf{H}}\mathbf{x}^0)$ denotes the residual vector.

Since our goal is to reduce the ML cost function in (\ref{costfunction}), when $\Delta \mathcal{F}_{\text{ML}}(d_i)< 0$, we update the $i$-th entry of the vector $\mathbf{x}^0$.
For the case of one symbol update, the maximum reduction in $\Delta \mathcal{F}_{\text{ML}}(d_i)$ can be achieved by forcing the gradient of (\ref{deltaF}) with respect to $d_i$ to zero.
Therefore, we have
\begin{equation}
	d_i  = \mathcal{P}\left(\frac{\tilde{e}_i}{\tilde{g}_{i,.i}}\right),
	\label{round}
\end{equation}
where $\mathcal{P}(a) = 2\lceil a/2 \rfloor$ represents the projection of $a$ onto the neighboring set $\mathcal{A}= \{0, \pm 2, \pm 4, \dots\}$.
If $\Delta \mathcal{F}_{\text{ML}}(d_i)< 0$, $d_i$ in (\ref{round}) indeed reduces the ML cost function and thus improves the performance; otherwise, the original solution is retained.
Defining $\tilde{\mathbf{D}} = \text{diag}(\tilde{\mathbf{G}}) \in \mathbb{R}^{K \times K}$, the update process can be expressed as
\begin{equation}
	\mathbf{x}^{\text{new}} = \mathbf{x}^0 + \mathcal{P}(\tilde{\mathbf{D}}^{-1} \odot \tilde{\mathbf{e}}).
\end{equation}
Since $\mathbf{x}^{\text{new}}_i$ may exceed $\sqrt{L} - 1$ or fall below $- \sqrt{L} + 1$, a projection step (\ref{projection}) is applied to adjust $\mathbf{x}^{\text{new}}_i$.
To sum up, the NS update process can be divided into the following two steps
\begin{equation}
	\mathbf{x}_{\text{NS}} = \mathbf{x}^0 + \tilde{\mathbf{D}}^{-1} \odot \tilde{\mathbf{e}}
\end{equation}
and
\begin{equation}
	\hat{\mathbf{x}} = \lceil \mathbf{x}_{\text{NS}} \rfloor_{\mathcal{Q}} \in \mathcal{X}^{K}.
\end{equation}
Here, $\hat{\mathbf{x}}$ is outputted as the detection solution of NS.
Then, subsitituting $\hat{\mathbf{x}}$ into \eqref{MILS}, the problem becomes
\begin{equation}
	 \hat{\mathbf{s}} = \underset{ \mathbf{s} \in \mathbb{R}^{M}} {\text{arg  
			min}} \left\| \mathbf{y}  - \mathbf{H}\hat{\mathbf{x}} - \mathbf{As} \right\|^2.
	\label{senresult}
\end{equation}
By solving \eqref{senresult} using the least squares (LS) method, the estimated sensing signal $\hat{\mathbf{s}}$ can be expressed as
\begin{equation}
	\hat{\mathbf{s}} = (\mathbf{A}^T\mathbf{A})^{-1}\mathbf{A}^T\left(\mathbf{y}-\mathbf{H}\hat{\mathbf{x}} \right).
\end{equation}

To summarize, the proposed projection-based NS-aided ADMM (P-NS-ADMM) detection for MIMO-ISAC systems, is outlined in Algorithm~\ref{P-NS-ADMM}.

\begin{algorithm}[t]
	\vspace{0.2cm}
	\renewcommand{\algorithmicrequire}{\textbf{Input:}}
	\renewcommand{\algorithmicensure}{\textbf{Output:}}
	\caption{Projection-based NS-aided ADMM (P-NS-ADMM) Detection for MIMO-ISAC Systems}
	\label{P-NS-ADMM}
	\begin{algorithmic}[1]
		\REQUIRE ${\bf{y}}$,  ${\bf{A}}$, ${\bf{H}}$,  $T_{\text{max}}$,  $\rho$
		\ENSURE communication signals $\hat{\bf{x}}$, sensing signals $\hat{\mathbf{s}}$
		\STATE \textbf{Initialize:} $\mathbf{P}=\mathbf{I}-\mathbf{A}(\mathbf{A}^T\mathbf{A})^{-1}\mathbf{A}^T$,  $\tilde{\mathbf{H}}=\mathbf{PH}$, $\tilde{\mathbf{y}}=\mathbf{Py}$, $\mathbf{z}^0 = \mathbf{0}$, ${\bm{\lambda}}^0 = {\bf 0}$
		\FOR {$k=1, 2, ..., T_{\text{max}}$}
		\STATE update ${{\bf{x}}^k}$ via (\ref{xupdate})
		\STATE update $\mathbf{z}^k$ via (\ref{zupdate})
		\STATE update $\bm{\lambda }^{k}$ via (\ref{lamopt})
		\ENDFOR
		\STATE set initial solution $\mathbf{x}^0_{\mathrm{}} = \lceil \mathbf{x}^{T_{\text{max}}}\rfloor_\mathcal{Q} \in \mathcal{X}^{K}$
		\STATE compute $\tilde{\mathbf{D}} = \text{diag}(\tilde{\mathbf{G}})$
		\STATE compute $\tilde{\mathbf{e}} = \tilde{\mathbf{y}} - \tilde{\mathbf{G}}\mathbf{x}^0$
		\STATE compute $\mathbf{x}_{\text{NS}} = \mathbf{x}^0 + \tilde{\mathbf{D}}^{-1} \odot \mathbf{e}$
		\STATE output $\hat{\mathbf{x}} =\lceil \mathbf{x}_{\text{NS}}\rfloor_\mathcal{Q}\in \mathcal{X}^{K}$
		\STATE output $\hat{\mathbf{s}} = \bm{\Xi} \mathbf{A}^T(\mathbf{y} - \mathbf{A}\hat{\mathbf{x}})$
	\end{algorithmic}
	\vspace{-0.1cm}
\end{algorithm}

\subsection{Complexity Analysis}
\label{secivc}

Here, the computational complexity is evaluated in terms of the number of real multiplications required, with each complex multiplication counted as four real multiplications \cite{10470369}.
The computational complexity of the P-NS-ADMM algorithm is comprised of three parts.
In the first part, the computational complexity of $\mathbf{P}$, $\tilde{\mathbf{H}}$, and  $\tilde{\mathbf{y}}$ is $M^3 + N^2M + 2M^2N$, $N^2K$, and $N^2$,  respectively.
Moreover, the calculation of $(\tilde{\mathbf{H}}^T\tilde{\mathbf{H}}+\rho \mathbf{I})^{-1}$ and $\tilde{\mathbf{H}}^{T}\tilde{\mathbf{y}}$ results in a complexity of $K^3 + K^2N$ and $NK$.
The second part involves ADMM iterations and the NS mechanism.
In the ADMM iterations, the computation of \eqref{xupdate} and \eqref{zupdate} requires $K^2+K$ and $K$ real multiplications, respectively.
As for the NS mechanism, the calculation of the residual vector $\tilde{\mathbf{e}}$ requires a complexity of $K^2$.
Subsequently, the calculation of $\mathbf{x}_{\mathrm{NS}}$, which involves a vector inversion and multiplication, demands a complexity of $2K$.
In the final part,  computing the estimated sensing signals $\hat{\mathbf{s}}$ requires a complexity of $MK + NK$.

To summarize, the computational complexity is given in Table~\ref{complexity}, where $T_{\text{max}}$ denotes the number of ADMM iterations,  $Q=\mathrm{log}_4L$ denotes the modulation order, and $T_\text{rand}$ represents the number of randomizations in the P-SDR algorithm.
For benchmarking, we integrate four state-of-the-art detectors into the proposed projection framework, where the prefix 'P-' marks the adaptation to this framework: 1) P-ADMIN \cite{shahabuddin2021admm}, employing the infinity-norm box constraint strategy; 2) P-PS-ADMM \cite{zhang2022efficient}, utilizing the consensus-based penalty sharing mechanism; and 3) P-SDR \cite{5370664}, applying semidefinite relaxation as a high-performance convex baseline.
These baselines share the identical preprocessing but apply their respective update rules to solve the problem \eqref{comdetection}.

\section{Detection Performance Analysis of P-NS-ADMM}

In this section, we show that the P-NS-ADMM algorithm achieves the same received diversity order as the optimal ML detection.

\subsection{The Post-projection Communication Channel}

We first derive the expression for the post-projection channel matrix $\tilde{\mathbf{H}}$ and then obtain the distribution of its entries.
\begin{lemma}
	Given the channel matrix $\mathbf{H}$ whose entries follow $\mathcal{N}(0, \frac{1}{2})$, as $N \to \infty$ and $N \gg M$, the entries of the post-projection channel matrix $\tilde{\mathbf{H}}$ i.i.d. follows
	\begin{equation}
		\tilde{h}_{i,j} \sim \mathcal{N}\left(0, \frac{1}{2}\left(1 - \frac{M}{N}\right)\right).
	\end{equation}
\end{lemma}
\begin{proof}
Based on the asymptotic orthogonality of steering vectors, the column vectors of ${\bar{\bf A}} = \left[ {{{\bf{a}}_r}({\theta _1}), \ldots ,{{\bf{a}}_r}({\theta _{M_t}})} \right] \in {\mathbb{C}^{{N_r} \times M_t}}$ are mutually orthogonal.
As the real-valued representation of $\bar{\mathbf{A}}$, the column vectors of $\mathbf{A}$ remain mutually orthogonal.

By defining $\mathbf{A} = [\mathbf{a}_1, \ldots, \mathbf{a}_{M}] \in \mathbb{R}^{N \times M}$, the product $\mathbf{A}^T\mathbf{A}$ takes the form
\begin{equation}
	\mathbf{A}^T\mathbf{A} = \text{diag}(\| \mathbf{a}_1\|^2, \ldots, \| \mathbf{a}_{M}\|^2).
\end{equation}
The projection matrix $\mathbf{P}$ is then given by
\begin{equation}
	\mathbf{P} = \mathbf{I} - \mathbf{A}(\mathbf{A}^T\mathbf{A})^{-1}\mathbf{A}^T = \mathbf{I} - \sum_{m=1}^{M} \frac{\mathbf{a}_m\mathbf{a}_m^T}{\| \mathbf{a}_m\|^2}.
\end{equation}
In this way,  the post-projection channel matrix can be expressed as
\begin{equation}
	\tilde{\mathbf{H}} = \mathbf{PH} = \mathbf{H} - \sum_{m=1}^{M} \frac{\mathbf{a}_m\mathbf{a}_m^T}{\| \mathbf{a}_m\|^2}\mathbf{H}.
\end{equation}
Given $\| \mathbf{a}_m\|^2=1$, the elements in $\tilde{\mathbf{H}}$ can be written as
\begin{equation}
	\tilde{h}_{i, j} = h_{i, j} - \sum_{m=1}^{M} a_{i, m}\sum_{n=1}^{N} a_{n, m}h_{n, j}. 
\end{equation}

To obtain the distribution of $\tilde{h}_{i, j}$, we analyze the mean and variance of $\tilde{h}_{i, j}$ and covariance matrix of $\tilde{\mathbf{H}}$.
Given the fact that $h_{i, j} \sim \mathcal{N}(0, \frac{1}{2})$, the mean of $\tilde{h}_{i, j}$ can be expressed by
\begin{equation}
	E[\tilde{h}_{i, j}] = E[h_{i, j}] - \sum_{m=1}^{M} a_{i, m}\sum_{n=1}^{N} a_{n, m}E[h_{n, j}] = 0,
\end{equation}
which results in
\begin{align}
	&Var[\tilde{h}_{i, j}] = E[\tilde{h}_{i, j}^2] 
	= E\left[\left(h_{i, j} - \sum_{m=1}^{M} a_{i, m}\sum_{n=1}^{N} a_{n, m}h_{n, j}\right)^2\right] \nonumber \\
	&=  E[h_{i,j}^2] - 2E\left[h_{i, j} \sum_{m=1}^{M} a_{i, m} \sum_{n=1}^{N} a_{n, m}h_{n, j}\right] \! \nonumber \\
	&+ E  \left[\sum_{p=1}^{M}\!\sum_{m=1}^{M}a_{i, p}a_{i, m}\!\sum_{q=1}^{N}\!\sum_{n=1}^{N}a_{q, p}a_{n, m}h_{q, j}h_{n, j}\right]. 
	\label{var}
\end{align}
Here, the first term is $E[h_{i,j}^2] = \sigma_h^2$ and the second term reduces to
\begin{align}
	&E\left[h_{i, j} \sum_{m=1}^{M} a_{i, m} \sum_{n=1}^{N} a_{n, m}h_{n, j}\right] \nonumber \\
	&=\! \!\sum_{m=1}^{M} a_{i, m}\sum_{n=1}^{N} a_{n, m} E[h_{i, j}h_{n, j}] = \sigma_h^2 \sum_{m=1}^{M} a_{i, m}^2.
\end{align}
Since $E[h_{i, j}h_{n, j}] \neq 0$ if and only if $i = n$, the last equality holds.
Similarly, the third term simplifies to
\begin{align}
	&E\left[\sum_{p=1}^{M}\sum_{m=1}^{M}a_{i, p}a_{i, m}\sum_{q=1}^{N}\sum_{n=1}^{N}a_{q, p}a_{n, m}h_{q, j}h_{n, j}\right] \nonumber \\
	&= \sum_{p=1}^{M}\!\sum_{m=1}^{M}a_{i, p}a_{i, m}\!\sum_{q=1}^{N}\!\sum_{n=1}^{N}a_{q, p}a_{n, m} E[h_{q, j}h_{n, j}]\nonumber \\
	&= \sigma_h^2\sum_{p=1}^{M}\!\sum_{m=1}^{M}a_{i, p}a_{i, m}(\mathbf{a}_{p}^T\mathbf{a}_{m})\nonumber \\
	&= \sigma_h^2 \sum_{m=1}^{M} a_{i, m}^2,
\end{align}
where the last equality holds due to $| \mathbf{a}_m|^2=1$ and the asymptotic orthogonality property.
Therefore, (\ref{var}) can be rewritten as
\begin{align}
	Var[\tilde{h}_{i, j}] &= \sigma_h^2 - 2 \sigma_h^2 \sum_{m=1}^{M} a_{i, m}^2 + \sigma_h^2 \sum_{m=1}^{M} a_{i, m}^2 \nonumber\\
	&= \sigma_h^2\left(1 -  \sum_{m=1}^{M} a_{i, m}^2\right) = \frac{1}{2}\left(1-\frac{M}{N}\right).
\end{align}

Finally, we derive the covariance of $\tilde{h}_{i_1, j_1}$ and $\tilde{h}_{i_2, j_2}$.
Given $E[\tilde{h}_{i, j}] = 0$, the covariance reduces to $Cov(\tilde{h}_{i_1, j_1}, \tilde{h}_{i_2, j_2}) = E[\tilde{h}_{i_1, j_1}\tilde{h}_{i_2, j_2}]$.
For $j_1 \neq j_2$, the mutual independence of zero-mean $h_{i,j}$ yields $E[\tilde{h}_{i_1, j_1}\tilde{h}_{i_2, j_2}] = 0$.
When $j_1 = j_2 = j$ with $i_1 \neq i_2$, the product $\tilde{h}_{i_1, j}\tilde{h}_{i_2, j}$ can be expressed as
\begin{align}
	&\tilde{h}_{i_1, j}\tilde{h}_{i_2, j} \!\nonumber\\
	&= - h_{i_1, j}\sum_{p=1}^{M} a_{i_2, p}\!\sum_{q=1}^{N} a_{q, p}h_{q, j} \!-\! h_{i_2, j}\sum_{m=1}^{M} a_{i_1, m}\!\sum_{n=1}^{N} a_{n, m}h_{n, j} \nonumber\\
	& \!+\! h_{i_1, j}h_{i_2, j} \!+\!\! \left(\sum_{m=1}^{M} a_{i_1, m}\!\sum_{n=1}^{N} a_{n, m}h_{n, j}\!\!\right)\!\! \left( \sum_{p=1}^{M} a_{i_2, p}\!\sum_{q=1}^{N} a_{q, p}h_{q, j}\!\!\right)\!\!,
\end{align}
which results in
\begin{align}
	&\mathbb{E}[\tilde{h}_{i_1, j}\tilde{h}_{i_2, j}] \nonumber\\
	&= -\frac{1}{2} \sum_{p=1}^{M} a_{i_2, p}\sum_{q=1}^{N} a_{q, p} \delta_{i_1,q}  - \frac{1}{2} \sum_{m=1}^{M} a_{i_1, m}\sum_{n=1}^{N} a_{n, m}\delta_{i_2,n}\nonumber\\
	&+ \frac{1}{2} \sum_{m=1}^{M}\sum_{p=1}^{M} a_{i_1, m}a_{i_2, p} \sum_{n=1}^{N}\sum_{q=1}^{N} a_{n, m}a_{q, p} \cdot \delta_{n,q}\nonumber\\
	&=-\sum_{m=1}^{M} a_{i_1, m}a_{i_2, m} + \frac{1}{2}\sum_{m=1}^{M}\sum_{p=1}^{M} a_{i_1, m}a_{i_2, p} (\mathbf{a}_m^T\mathbf{a}_p)\nonumber\\
	&= -\sum_{m=1}^{M} a_{i_1, m}a_{i_2, m} + \frac{1}{2}\sum_{m=1}^{M} a_{i_1, m}a_{i_2, m} \nonumber\\
	&= -\frac{1}{2}\sum_{m=1}^{M} a_{i_1, m}a_{i_2, m} \approx 0.
	\label{cov}
\end{align}
In (\ref{cov}), since $\sum_{m=1}^{M}a_{i_1, m}a_{i_2, m} < M/N$, the last approximation holds under the assumptions of $N \gg M$.
In summary, we can obtain $\tilde{h}_{i, j} \sim \mathcal{N}\left(0, \frac{1}{2}\left(1 - \frac{M}{N}\right)\right)$, thus completing the proof.
\end{proof}

\subsection{Equivalent Noise}

In the P-NS-ADMM algorithm, we obtain $\mathbf{x}^{T_{\text{max}}}$ through ADMM iterations and then $\mathbf{x}^{T_{\text{max}}}$ is porjected onto the constellation set $\mathcal{X}^{K}$ as an initial input for the NS mechanism, which can be expressed as 
\begin{align}
	\mathbf{x}^0 &= \lceil \mathbf{x}^{T_{\text{max}}}  \rfloor_{\mathcal{Q}} \nonumber\\
	&= \mathbf{x} + \mathbf{n}^0 + \bm{\delta}\nonumber\\
	&= \mathbf{x} + \mathbf{w}.
	\label{x0}
\end{align}
Here $\mathbf{x}^0 \in \mathbb{R}^{K}$ denotes the initial input and $\mathbf{n}^0 \in \mathbb{R}^{K}$ denotes the equivalent noise vector from the ADMM iterations.
Additionally, $\mathbf{w} \in \mathcal{A}^{K}$ is the distance vector from the neighboring set, and $\bm{\delta} = \mathbf{w} - \mathbf{n}^0 \in \mathbb{R}^{K}$ rounds the solution to the nearest constellation point.
Clearly, the elements in $\mathbf{w}$ are distributed as follows:
\begin{equation}
	w_i =
	\begin{cases} 
		0, & \text{if } |\mathbf{n}^0_i| \leq 1, \\
		2, & \text{if } 1 < \mathbf{n}_i^0 \leq 3, \\
		-2, & \text{if } -3 \leq  \mathbf{n}_i^0 < -1,\\
		\cdots &\cdots
	\end{cases}
\end{equation}
for $i = 1, 2, \ldots, K$.
Substituting (\ref{x0}) into the residual vector $\tilde{\mathbf{e}} = \tilde{\mathbf{H}}^T(\tilde{\mathbf{y}} - \tilde{\mathbf{H}}\mathbf{x}_0)$, we can obtain
\begin{align}
	\tilde{\mathbf{e}} &= \tilde{\mathbf{H}}^T(\tilde{\mathbf{y}} - \tilde{\mathbf{H}}\mathbf{x}_0)\nonumber \\
	&= \tilde{\mathbf{H}}^T(\tilde{\mathbf{H}}\mathbf{x} +\tilde{\mathbf{n}} -\tilde{\mathbf{H}}(\mathbf{x} + \mathbf{w}) )\nonumber\\
	&= \tilde{\mathbf{H}}^T\tilde{\mathbf{n}} - \tilde{\mathbf{H}}^T\tilde{\mathbf{H}}\mathbf{w},
\end{align}
where $\tilde{\mathbf{n}}$ is the post-projection noise.
Thus, through the update of NS, the solution $\mathbf{x}_{\text{NS}}$ is given by
\begin{align}
	\mathbf{x}_{\text{NS}} &= \mathbf{x}^0 + \tilde{\mathbf{D}}^{-1} \odot \mathbf{e} \nonumber\\
	&= \mathbf{x} + \mathbf{w} + \tilde{\mathbf{D}}^{-1} \odot (\tilde{\mathbf{H}}^T\tilde{\mathbf{n}} - \tilde{\mathbf{H}}^T\tilde{\mathbf{H}}\mathbf{w}) \nonumber \\
	&= \mathbf{x} + \mathbf{w} + \tilde{\mathbf{D}}^{-1} \odot (\tilde{\mathbf{H}}^T\tilde{\mathbf{n}}) - \tilde{\mathbf{D}}^{-1} \odot (\tilde{\mathbf{H}}^T\tilde{\mathbf{H}}\mathbf{w}),
	\label{xns}
\end{align}
which will be projected onto the constellation set $\mathcal{X}^{K}$ as the final detection output.
We now analyze the equivalent noise of the P-NS-ADMM algorithm.
\begin{lemma}
Given the channel matrix $\tilde{\mathbf{H}}$ whose entries follow $\mathcal{N}\left(0, \frac{1}{2}\left(1 - \frac{M}{N}\right)\right)$, the distance vector $\mathbf{w}$ follows that
\label{enoise}
\begin{equation}
	\mathbf{w} = \tilde{\mathbf{D}}^{-1} \odot (\tilde{\mathbf{H}}^T\tilde{\mathbf{H}}\mathbf{w})
\end{equation}
with the increment of SNR.
\end{lemma}
\begin{proof}
In the case of $i=j$, since $\tilde{h}_{i, j}^2$ follows the  Gamma distribution with the shape parameter $\alpha=\frac{1}{2}$ and the scale parameter $\gamma=1-\frac{M}{N}$, the summation $\tilde{g}_{i, i} = \sum_{n = 1}^{N}\tilde{h}_{n, i}^2$ follows:
\begin{equation}
	\tilde{g}_{i, i} \sim \Gamma \left(\frac{N}{2}, 1 - \frac{M}{N} \right).
	\label{gii}
\end{equation}
For $i \neq j$, $\tilde{h}_{n, i}$ and $\tilde{h}_{n, j}$ are independent Gaussian variables with mean zero and variance $\frac{1}{2}(1 - \frac{M}{N})$, resulting in
\begin{equation}
	\tilde{g}_{i, j} = \sum_{n = 1}^{N}\tilde{h}_{n, i}\tilde{h}_{n, j} \sim \mathcal{N}\left(0, \frac{N}{4}\left(1-\frac{M}{N}\right)^2\right).
	\label{gij}
\end{equation}

Then we define
\begin{equation}
	\mathbf{b} = \mathbf{D}^{-1} \odot (\tilde{\mathbf{H}}^T\tilde{\mathbf{H}}\mathbf{w}) = (\text{diag}(\mathbf{G}))^{-1} \odot (\mathbf{Gw})
\end{equation}
with the $i$-th element
\begin{equation}
	b_i = w_i + \frac{1}{\tilde{g}_{i, i}}\sum_{j \neq i}^{K} \tilde{g}_{i, j} w_j.
\end{equation}
Conditioned on $\tilde{g}_{i,i}$, $b_i$ represents a linear combination of independent Gaussian variables $\tilde{g}_{i,j}$. 
Consequently, $b_i$ follows a conditional Gaussian distribution expressed as
\begin{equation}
	b_i \sim \mathcal{N}\left(w_i, \frac{(N - M)^2}{4Ng_{i, i}^2}\sum_{j \neq i}^{K}w_j^2 \right).
	\label{bi}
\end{equation}
Given $\tilde{g}_{i, i} \sim \Gamma\left(\frac{N}{2}, 1 - \frac{M}{N}\right)$, its mean and variance can be simply expressed as
\begin{align}
	&E[\tilde{g}_{i, i}] = \frac{N}{2}\left(1 - \frac{M}{N}\right) = \frac{N-M}{2}, \nonumber\\
	&Var[\tilde{g}_{i, i}] = \frac{N}{2} \left(1 - \frac{M}{N}\right)^2 = \frac{(N - M)^2}{2N},
\end{align}
which leads to
\begin{equation}
	E[\tilde{g}_{i, i}^2]= E^2[\tilde{g}_{i, i}] + Var[\tilde{g}_{i, i}] = \frac{(N + 2)(N - M)^2}{4N}.
\end{equation}

Taking the expectation of $\tilde{g}_{i, i}^2$  in (\ref{bi}), the variance of $b_i$ can be approximated as
\begin{align}
	\sigma_{b_i}^2 &\approx \frac{(N - M)^2}{4N} \frac{4N}{(N + 2)(N - M)^2}(K -1)E[w_j^2 ] \nonumber \\
	&= \frac{K - 1}{N + 2}E[w_j^2 ] < \varepsilon,
	\label{bivar}
\end{align}
where $\varepsilon$ denotes an arbitrarily small positive number. 
This inequality holds given the assumptions of $N \gg K$ and high SNR regimes.
Under these conditions, the variance $\sigma_{b_i}^2$ approaches zero, justifying the approximation $b_i \approx E[b_i]$ \cite{chen2025decentralized}, thus completing the proof.
\end{proof}

Based on Lemma~\ref{enoise}, the NS solution $\mathbf{x}_{\text{NS}}$ in (\ref{xns}) simplifies to
\begin{equation}
	\mathbf{x}_{\text{NS}} = \mathbf{x} + \tilde{\mathbf{D}}^{-1} \odot (\tilde{\mathbf{H}}^T\tilde{\mathbf{n}}).
	\label{xns_new}
\end{equation}
Therefore, given the post-projection channel matrix $\tilde{\mathbf{H}}$, the equivalent noise of the P-NS-ADMM algorithm is
\begin{equation}
	\mathbf{n}_{\text{NS}} = \tilde{\mathbf{D}}^{-1} \odot (\tilde{\mathbf{H}}^T\tilde{\mathbf{n}}).
	\label{equnoise}
\end{equation}

\subsection{Received Diversity Order}
In this subsection, we derive the post-processing SNR distribution for P-NS-ADMM, which enables subsequent characterization of the received diversity order.

\begin{lemma}
Given the equivalent noise in \eqref{equnoise}, the postprocessing SNR $\gamma_i$ on the $i$-th element of the P-NS-ADMM algorithm is
\begin{equation}
	\gamma_i=\gamma_0 \tilde{g}_{i, i},
\end{equation}
where $\gamma_0=E[x_i^2]/ \sigma^2_n$ denotes the average SNR on each element.
\end{lemma}
\begin{proof}
We begin by analyzing the equivalent noise term in (\ref{equnoise}), where the $i$-th component is expressed as
\begin{equation}
	\tilde{n}_i=\frac{\tilde{\mathbf{h}}_i^T \tilde{\mathbf{n}}}{\tilde{g}_{i, i}}.
\end{equation}
In this way, the corresponding noise power can be derived as
\begin{align}
	\sigma_i^2&=\frac{E\left[(\tilde{\mathbf{h}}_i^T \tilde{\mathbf{n}})^2\right]}{\tilde{g}_{i, i}^2}\nonumber\\
	&=\frac{E[\mathrm{Tr}(\tilde{\mathbf{h}}_i^T \tilde{\mathbf{n}}\tilde{\mathbf{n}}^T \tilde{\mathbf{h}}_i )]}{{\tilde{g}_{i, i}^2}}\nonumber\\
	&=\frac{\mathrm{Tr}( \tilde{\mathbf{h}}_i^T E[\tilde{\mathbf{n}}\tilde{\mathbf{n}}^T] \tilde{\mathbf{h}}_i)}{\tilde{g}_{i, i}^2}\nonumber\\
	&= \frac{\mathrm{Tr}( \tilde{\mathbf{h}}_i^T \sigma_n^2 \mathbf{P} \tilde{\mathbf{h}}_i)}{\tilde{g}_{i, i}^2}\nonumber\\
	&= \frac{\sigma_n^2 \mathrm{Tr}(\tilde{\mathbf{h}}_i^T\tilde{\mathbf{h}}_i)}{\tilde{g}_{i, i}^2}\nonumber\\
	&= \frac{\sigma_n^2}{\tilde{g}_{i, i}}.
\end{align}
Consequently, we obtain the post-processing SNR for the $i$-th element as follows
\begin{equation}
	\gamma_i = \frac{E[x_i^2]}{\sigma_i}=\gamma_0 \tilde{g}_{i, i},
	\label{gammai}
\end{equation}
thus completing the proof.
\end{proof}
\begin{theorem}
The post-processing SNR for the $i$-th element in the P-NS-ADMM algorithm follows a weighted Chi-square distribution, expressed as
\begin{equation}
		f(\gamma_i) =\frac{e^{-\frac{\gamma_i}{\gamma_0(1-M/N)}}}{\Gamma(\frac{N}{2})\gamma_0(1-M/N)}\!\left( \frac{\gamma_i}{\gamma_0(1-M/N)}\right)^{\frac{N}{2}-1}\!\!,
\end{equation}
\end{theorem}
\begin{proof}
Building upon the Gamma-distributed variable $\tilde{g}_{i, i}$ from (\ref{gii}), we introduce a transformed variable $v=\frac{2\tilde{g}_{i, i}}{1-M/N}$, which follows a Chi-square distribution 
\begin{equation}
	v \sim \chi^2(N)
\end{equation}
with $N$ degrees of freedom.
Consequently, PDF of $v$ is given by
\begin{equation}
	f(v)=\frac{e^{-\frac{v}{2}}}{2^{\frac{N}{2}}\Gamma(\frac{N}{2})}v^{\frac{N}{2}-1}.
\end{equation}
Substituting $\tilde{g}_{i, i}=\frac{v(1-M/N)}{2}$ into  (\ref{gammai}), $\gamma_i$ can be expressed as
\begin{equation}
	\gamma_i = \bar{\gamma}_0 v,
\end{equation}
where $\bar{\gamma}_0=\gamma_0\frac{(1-M/N)}{2}$.
According to the variable transformation formula, the exact PDF of the post-processing SNR is derived as
\begin{align}
	f(\gamma_i)&=f_v\left(\frac{\gamma_i}{\bar{\gamma}_0}\right) \times \left(\frac{d}{d\gamma_i}\left(\frac{\gamma_i}{\bar{\gamma}_0}\right)\right) \nonumber\\
	&= f_v\left(\frac{\gamma_i}{\bar{\gamma}_0}\right) \times \frac{1}{\bar{\gamma}_0} \nonumber \\
	&= \frac{e^{-\frac{\gamma_i}{2\bar{\gamma}_0}}}{2^{\frac{N}{2}}\Gamma(\frac{N}{2})}\left(\frac{\gamma_i}{\bar{\gamma}_0}\right)^{\frac{N}{2}-1}\nonumber\\
	&=\frac{e^{-\frac{\gamma_i}{\gamma_0(1-M/N)}}}{\Gamma(\frac{N}{2})\gamma_0(1\!-\! M/N)}\left( \frac{\gamma_i}{\gamma_0(1 \!-\! M/N)}\right)^{\frac{N}{2}-1}.
\end{align}
\end{proof}
This demonstrates that the post-processing SNR $\gamma_i$ follows a weighted Chi-square distribution with $N$ degrees of freedom.
Consequently, we establish the following Theorem, proving that P-NS-ADMM can achieve the full received diversity order.

\begin{theorem}
\label{full_diversity}
The achievable received diversity order of the P-NS-ADMM algorithm is
\begin{equation}
	d_{\mathrm{NS}}=\lim_{\sigma_n^2 \to 0}\frac{\mathrm{log}\left(\bar{P}_{e, i}\right)}{\mathrm{log}(\sigma_n^2)}=N_r.
\end{equation}
\end{theorem}
\begin{proof}
Under the assumption of independent ML decoding, the PER for the $i$-th element can be approximated as \cite{paulraj2003introduction}
\begin{equation}
	P_{e, i} \approx N_e Q(\sqrt{2\gamma_i})\leq \frac{1}{2}N_e e^{-\gamma_i},
\end{equation}
where $N_e$ represents the average number of nearest-neighbor constellation points for the $i$-th stream, with the inequality derived from the Chernoff bound, and $Q(\cdot)$ is the standard $Q$-function.

Given that the channel experiences fading across the $N_r$ correlated links, the received diversity order is determined by the expectation of the PER (i.e., $\bar{P}_{e, i}=E[P_{e, i}]$) \cite{gore2002transmit}, which can be expressed as
\begin{equation}
	\bar{P}_{e, i} \leq \frac{N_e}{2}\left(1 + \gamma_0\left(1-\frac{M}{N}\right)\right)^{-N_r}.
\end{equation}
Thus, the received diversity order of the P-NS-ADMM algorithm is lower-bounded by
\begin{align}
	d_{\mathrm{NS}}&=\lim_{\sigma_n^2 \to 0}\frac{\mathrm{log}\left(\bar{P}_{e, i}\right)}{\mathrm{log}(\sigma_n^2)}\nonumber\\
	&\geq \lim_{\sigma_n^2 \to 0} \frac{\mathrm{log}\left(\frac{N_e}{2}\right) \!-\! N_r \mathrm{log} \left(1 \!+\! \frac{E[x_i^2]}{\sigma^2_n}\left(1 \! - \!\frac{M}{N}\right)\right)}{\mathrm{log}(\sigma_n^2)}\nonumber\\
	&\approx  \lim_{\sigma_n^2 \to 0} \frac{\mathrm{log}\left(\frac{N_e}{2}\right) \!-\! N_r \mathrm{log} (E[x_i^2]) \!-\! N_r \mathrm{log} \left(1 \!- \!\frac{M}{N}\right)  }{\mathrm{log}(\sigma_n^2)} \!+\! N_r\nonumber\\
	&=N_r.
	\label{lowerbound}
\end{align}
Since ML detection provides the  performance upper bound, the achievable received diversity order of P-NS-ADMM satisfies
\begin{equation}
	d_{\text{NS}} \leq d_{\text{ML}} = N_r,
\end{equation}
which coincides with the lower bound in (\ref{lowerbound}), resulting in  $d_{\text{NS}}= N_r$.
\end{proof}
As can be seen clearly, the proposed P-NS-ADMM algorithm attains full received diversity order, demonstrating its ability to asymptotically achieve the optimal ML detection performance as SNR increases.
Moreover, it is worth noting that the last equation in \eqref{lowerbound} holds only if $N \gg M$. 
Consequently, a performance degradation becomes inevitable as $M$ approaches $N$.

\section{The Proposed Iteration-based NS-aided ADMM Detection Scheme}
\label{secvi}

We find that the projection operation in P-NS-ADMM is computationally expensive.
To this end, we propose an iteration-based detection scheme that avoids the projection operation, thereby significantly reducing the complexity.

We note that MIMO-ISAC signal detection comprises two subproblems: communication detection and sensing estimation.
The NS-aided ADMM approach is suitable for this problem, as it can iteratively solve these sub-problems while updating parameters alternately.
Specifically, unlike the projection-based scheme, here we directly introduce auxiliary variable $\mathbf{z}$ in \eqref{MILS} and relax it to the box constraints.
The corresponding problem can be expressed as
\begin{align}
	\underset {\mathbf{x} \in \mathbb{R}^{K}, \mathbf{s} \in \mathbb{R}^{M}, \mathbf{z} \in \mathcal{R}^{K}_{\mathcal{X}}} {\text{ min}} \quad & \| \mathbf{y} - \mathbf{Hx} -\mathbf{As} \|^2 , \nonumber \\
	\text{s.t.} \quad & \mathbf{x} = \mathbf{z}.
	\label{relax2}
\end{align}
Similarly, by solving the optimization problem \eqref{relax2}, the update of $\mathbf{x}$ and $\mathbf{s}$ are given by
\begin{equation}
	\label{xupdate2}
	\hspace{-0.2cm}
	\mathbf{x}^{k}\!=\!(\mathbf{H}^T\mathbf{H}\!+\!\rho\mathbf{I})^{-1}\!\left(\mathbf{
		H}^T\mathbf{y}\!-\!\mathbf{H}^T\mathbf{As}^{k-1} \!\!+\! \rho(\mathbf{z}^{k-1} \!-\! \bm{\lambda}^{k-1})\right)\!,
\end{equation}
\begin{equation}
	\label{supdate}
	\mathbf{s}^{k}=(\mathbf{A}^T\mathbf{A})^{-1}\mathbf{A}^T(\mathbf{y}-\mathbf{Hx}^{k}).
\end{equation}
Since the constraints in \eqref{relax2} are independent of $\mathbf{s}$, the update expression of $\mathbf{z}$ and $\bm{\lambda}$ remain the same, i.e., \eqref{zupdate} and \eqref{lamopt}.
As for the NS mechanism, it still works here by simply replacing $\mathbf{y}$ in \eqref{costfunction} with $\mathbf{y}-\mathbf{As}$.
Meanwhile, to speed up the convergence of the ADMM algorithm, $\mathbf{z}$ and $\bm{\lambda}$ are updated twice in a single iteration.
Finally, the detection results $\hat{\mathbf{x}}$ and $\hat{\mathbf{s}}$ are output when the algorithm reaches the maximum number of iterations.
To summarize, the proposed iteration-based NS-aided ADMM (I-NS-ADMM) detection is outlined in Algorithm~\ref{I-NS-ADMM}.

\begin{algorithm}[t]
	\vspace{0.2cm}
	\renewcommand{\algorithmicrequire}{\textbf{Input:}}
	\renewcommand{\algorithmicensure}{\textbf{Output:}}
	\caption{Iteration-based NS-aided ADMM (I-NS-ADMM) Detection for MIMO-ISAC Systems}
	\label{I-NS-ADMM}
	\begin{algorithmic}[1]
		\REQUIRE ${\bf{y}}$,  ${\bf{A}}$, ${\bf{H}}$,  $T_{\mathrm{max}}$,  $\rho$
		\ENSURE communication signals $\hat{\bf{x}}$, sensing signals $\hat{\mathbf{s}}$
		\STATE \textbf{Initialize:} $\mathbf{s}^0 = {\mathbf 0}$, $\mathbf{z}^0 = {\mathbf 0}$, ${\bm{\lambda}}^0 = {\bf 0}$
		\FOR {$k=1, 2, ..., T_{\text{max}}$}
		\STATE update ${{\bf{x}}^{k}}$ via (\ref{xupdate2}), 
		\STATE update ${{\bf{s}}^{k}}$ via (\ref{supdate}), 
		\STATE update $\mathbf{z}^{k}$ via (\ref{zupdate})
		\STATE update $\bm{\lambda }^{k}$ via (\ref{lamopt})
		\STATE set NS initial solution $\mathbf{x}^0 = \lceil \mathbf{x}^{k}\rfloor_\mathcal{Q} \in \mathcal{X}^{K}$
		\STATE compute $\mathbf{D} = \text{diag}(\mathbf{G})$
		\STATE compute $\mathbf{e} = \mathbf{y} -\mathbf{As}- \mathbf{G}\mathbf{x}^0$
		\STATE compute $\mathbf{x}_{\mathrm{NS}} = \mathbf{x}^0 + \mathbf{D}^{-1} \odot \mathbf{e}$
		\STATE compute $\mathbf{x}^{k}=\lceil \mathbf{x}_{\text{NS}}\rfloor_\mathcal{Q}\in \mathcal{X}^{K}$
		\STATE update $\mathbf{z}^{k}$ via (\ref{zupdate})
		\STATE update $\bm{\lambda }^{k}$ via (\ref{lamopt})
		\ENDFOR
		\STATE output $\hat{\mathbf{x}} = \mathbf{x}^{T_{\mathrm{max}}}$
		\STATE output $\hat{\mathbf{s}}=\mathbf{s}^{T_{\mathrm{max}}}$ 
	\end{algorithmic}
	\vspace{-0.1cm}
\end{algorithm}

From Algorithm~\ref{I-NS-ADMM}, it is clear that the estimation of $\mathbf{s}$ is determined by ADMM iterations.
However, in systems where $N$ is close to $M+K$, due to the slow convergence of ADMM, an accurate estimation of $\mathbf{s}$ may not be readily obtained with only a single ADMM iteration, which subsequently degrades the communication detection performance.
To mitigate this issue, we introduce a flexible ADMM iteration mechanism that improves the estimation of $\mathbf{s}$ by increasing iteration counts.
Specifically, as depicted in Fig.~\ref{flexible}, the mechanism performs NS every $I$ times ADMM iterations, repeating the process up to $P$ times.
On the other hand, due to the absence of the projection operation, mutual interference between S\&C signals is inevitable, resulting in a slight performance loss, as demonstrated by the simulation results.
\begin{figure*}[t]
	\centering
	\includegraphics[width=0.73\textwidth]{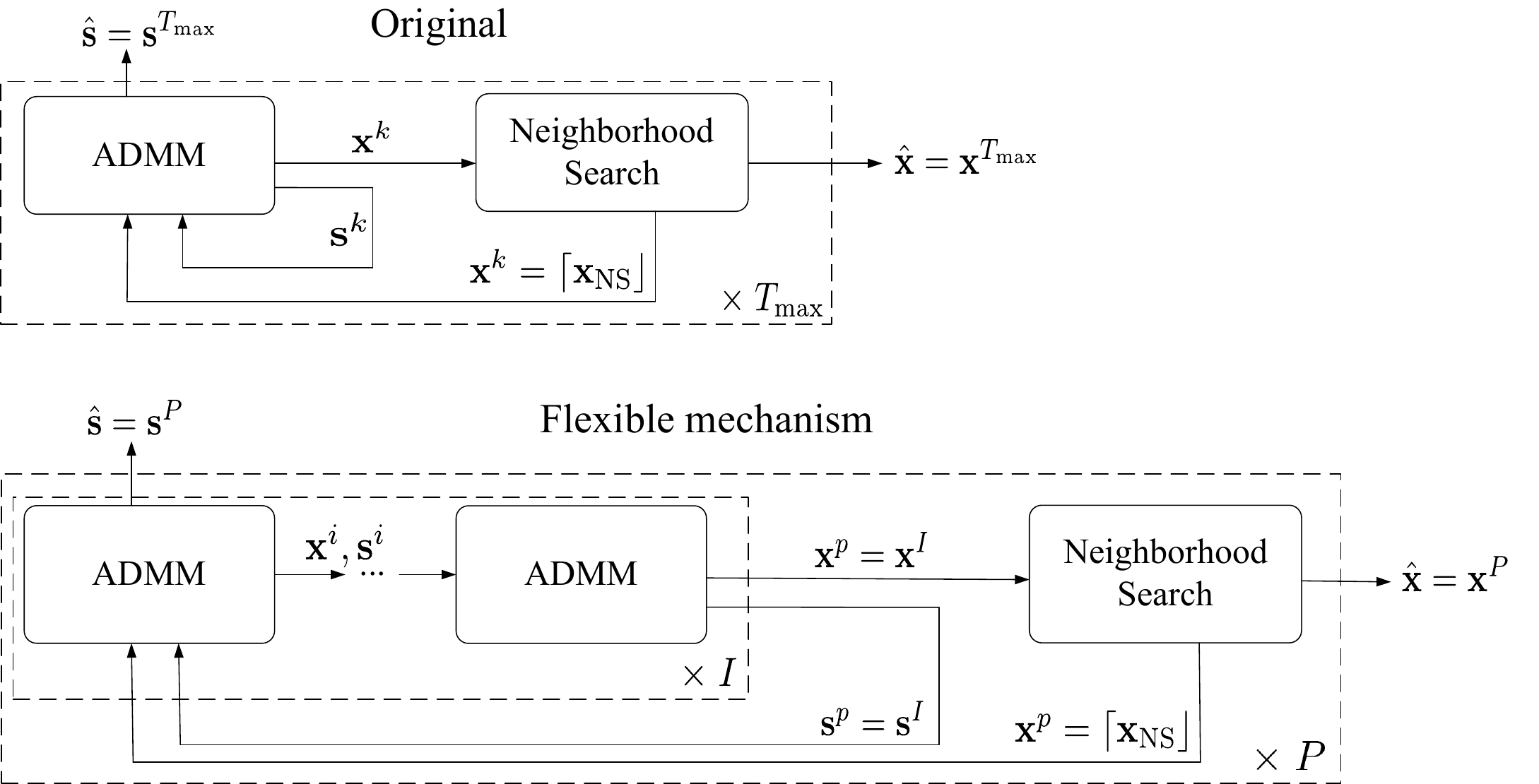}
	\caption{The flexible mechanism of the ADMM iterations.}
	\label{flexible}
\end{figure*}

To visually present the complexity advantages of the iteration-based scheme, we analyze the overall computational complexity of I-NS-ADMM, which consists of preprocessing and iteration.
In the preprocessing part, the I-NS-ADMM algorithm also requires the calculation of $(\mathbf{H}^T\mathbf{H}+\rho \mathbf{I})^{-1}$, $\mathbf{H}^{T}\mathbf{y}$, and $(\mathbf{A}^{T}\mathbf{A})^{-1}\mathbf{A}^{T}$.
Moreover, an additional computation of $\mathbf{H}^{T}\mathbf{A}$ is required, which involves $NMK$ real multiplications.
However, the main difference is that the I-NS-ADMM algorithm has no step to compute the projection matrix and the post-projection variants, significantly reducing the complexity.
In the iteration part, the computational burdens for determining $\mathbf{x}^{k}$, $\mathbf{s}^{k}$, and $\mathbf{z}^{k}$ in each iteration are noted as $K^2 + KM + K$, $NM + NK$, and $2K$.
Since NS is independent of the sensing signals, the complexity here is also $K^2 + 2K$.

To summarize, the overall computational complexity is illustrated in Table~\ref{complexity}.
Here, the prefix 'I-' denotes the iteration-based schemes, which follow the alternating structure of Algorithm~\ref{I-NS-ADMM} but employ the solvers from \cite{shahabuddin2021admm} and \cite{zhang2022efficient} for the communication detection step.
Note that SDR is excluded from this framework as its specific structures are unsuitable for the alternating iterative updates.
Comparing the two schemes, the complexity of projection-based schemes is predominantly concentrated in the preprocessing stage, specifically $N^2K$ and $N^2M$.
Meanwhile, iteration-based schemes exhibit higher complexity in the iteration part, specifically $T_{\text{max}}NK$ and $T_{\text{max}}NM$, and $NKM$ in the preprocessing part.
Consequently, when $T_{\text{max}}<N-KM/(K+M)$, iteration-based schemes have significant complexity advantages.

\begin{table*}[t]
	\renewcommand\arraystretch{1.7}
	\begin{center}
		\caption{Computational Complexity Comparison of Different Algorithms}
		\label{complexity}
		\vspace{-0.1cm}
		{\begin{threeparttable}
				\begin{tabular}{|c|c|c|}
					\hline
					\textbf{Algorithm} & \textbf{Number of real multiplications without iteration part} & \textbf{Iteration / Solving part}\\
					\hline
					P-ADMIN & \multirow{3}{*}{$N^2K \!+\! N^2M \!+\! N^2 \!+\! NK^2 \!+\! 2NM^2 \!+\! 2NK \!+\! K^3 \!+\! KM \!+\! M^3$} & $T_{\text{max}}(K^2 \!+\! 2K)$\\
					\cline{1-1}\cline{3-3}
					P-PS-ADMM & & $T_{\text{max}}(K^2 \!+\! KQ^2 \!+\! KQ)$ \\
					\cline{1-1}\cline{3-3}
				P-SDR & & $\mathcal{O}(K^{3.5}) + T_{\text{rand}}NK$ \\
					\hline
					P-NS-ADMM & $N^2K \!+\! N^2M \!+\! N^2 \!+\! NK^2 \!+\! 2NM^2 \!+\! 2NK \!+\! K^3 \!+\! K^2 \!+\! KM \!+\! 2K \!+\! M^3\!$ & $T_{\text{max}}(K^2+ 2K)$\\
					\hline
					I-ADMIN& \multirow{3}{*}{$NK^2 \!+\! 2NM^2 \!+\! NKM \!+\! NK \!+\! K^3 \!+\! M^3$} & $T_{\text{max}}(NK \!+\! NM \!+\! K^2 \!+\! KM \!+\! 3K)$\\
					\cline{1-1}\cline{3-3}
					I-PS-ADMM & & $T_{\text{max}}(NK \!+\! NM \!+\! K^2 \!+\! KM \!+\! 2KQ^2 \!+\! KQ)$  \\
					\cline{1-1}\cline{3-3}
					I-NS-ADMM &  & $T_{\text{max}}(NK \!+\! NM \!+\! 2K^2 \!+\! KM \!+\! 5K)$\\
					\hline
				\end{tabular}
		\end{threeparttable}}
		\vspace{-0.2cm}
	\end{center}
\end{table*}

\section{Simulation}

In this section, numerical results are presented to evaluate the performance of the proposed NS-aided ADMM detection in uplink large-scale MIMO-ISAC systems.
We benchmark our proposals against the projection-based (P-ADMIN, P-PS-ADMM, P-SDR) and iteration-based (I-ADMIN, I-PS-ADMM) schemes introduced in Section~\ref{secivc} and Section~\ref{secvi}.
Note that we restrict our comparison to model-driven approaches to validate the proposed theoretical propositions, while data-driven deep learning methods are left for future investigation.
In terms of communication performance evaluation, we first compare the bit error rate (BER) performance.
Next, we analyze the computational complexity of all algorithms under different scenarios with varying iteration counts.
Finally, we contrast P-NS-ADMM with the ML detection to demonstrate the full received diversity order as established in Theorem~\ref{full_diversity}.
Regarding sensing performance, the normalized mean square error (NMSE), between the estimated reflection coefficient $\hat {\beta} $ and the real $\beta $ is used as a performance metric here.
We evaluate the NMSE performance under different scenarios and derive a theoretical benchmark as a reference for performance analysis.
Unless otherwise stated, the simulation setup follows the system model defined in Section~\ref{sec2}.
For brevity in the figure captions, we use the notation $N_r \times U \times M_t$ to denote a system configuration with $N_r$ BS antennas, $U$ single-antenna users, and $M_t$ single-antenna targets.
In addition, the power of the communication signals ${P_c}$ is set to $1$ and the power ratio between the communication signals and sensing signals ${P_c}/{P_s}$ is set to $-8~\mathrm{dB}$.

\vspace{-0.1cm}
Fig.~\ref{com_128_32_2_and_128_32_24} presents a comparison of communication detection performance in the $32 \times 8 \times 2$ and $32 \times 8 \times 8$ MIMO-ISAC systems with 16-QAM.
The number of ADMM iterations for all algorithms is set to $T_{\text{max}}=10$.
As demonstrated in the figure, the NS-aided algorithms exhibit superior BER performance.
In addition, the performance of all algorithms shows a consistent decline as the number of sensing targets grows, which is consistent with the analysis of \eqref{lowerbound}.

\begin{figure}[t]
	\centering
	\vspace{-0.3cm}
	\includegraphics[width=0.48\textwidth]{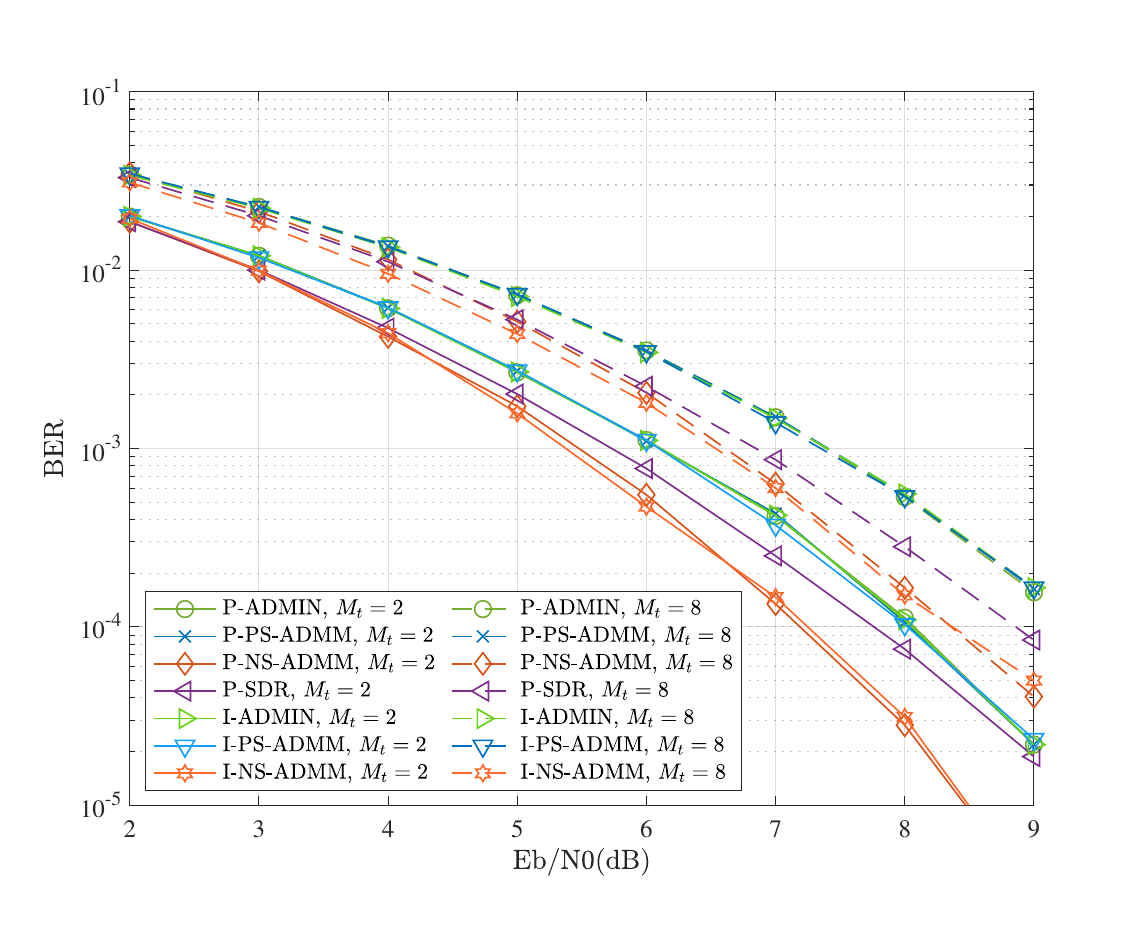}
	\vspace{-0.5cm}
	\caption{Comparisons of BER performance for 16-QAM in $32 \times 8 \times 2$ and $32 \times 8 \times 8$ MIMO-ISAC systems.}
	\label{com_128_32_2_and_128_32_24}
	\vspace{-0.3cm}
\end{figure}

Fig.~\ref{com_64_60_2} extends the BER performance comparison to a $64 \times 60 \times 2$ MIMO-ISAC system with 16-QAM.
The number of ADMM iterations is set to $10$ and $20$, and the flexible mechanism's iteration parameter $I$ is set to $2$.
On the one hand, the iteration-based scheme outperforms the projection-based scheme at low SNR.
This is due to the fact that in this system, where $N_r$ is close to $M_t + U$, the post-projection channel matrix is more likely to turn into an ill-conditioned matrix, which affects communication detection.
Moreover, I-NS-ADMM with the flexible mechanism demonstrates better performance at low SNR compared with I-NS-ADMM.
On the other hand, since the iteration-based schemes are affected by sensing signals, the detection performance will reach the upper bound at high SNR.

\begin{figure}[t]
	\vspace{-0.3cm}
	\centering
	\includegraphics[width=0.48\textwidth]{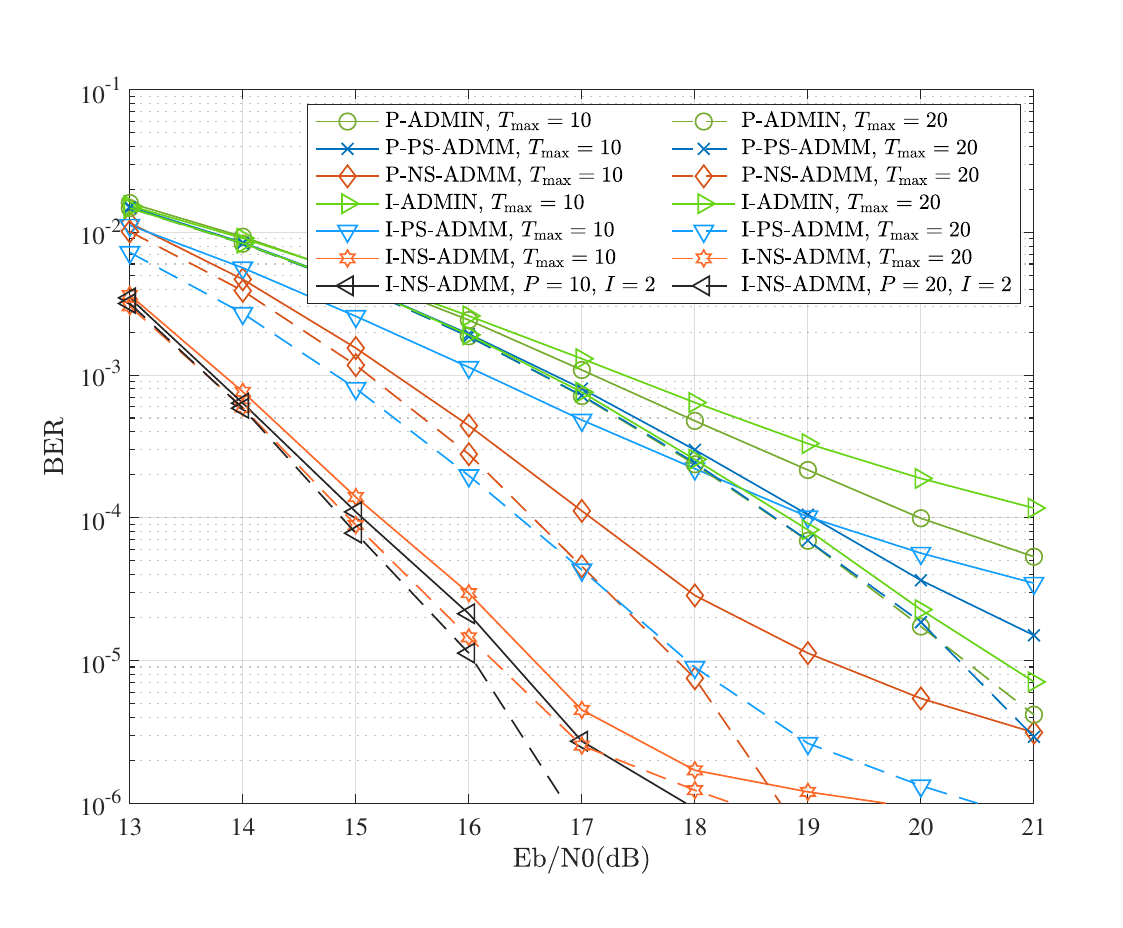}
	\vspace{-0.5cm}
	\caption{Comparisons of BER performance for 16-QAM in $64 \times 60 \times 2$ MIMO-ISAC system.}
	\vspace{-0.3cm}
	\label{com_64_60_2}
\end{figure}


As illustrated in Fig.~\ref{complexity_comparison}, the complexity comparison of P-PS-ADMM, P-NS-ADMM, I-PS-ADMM, and I-NS-ADMM is demonstrated for three scenarios in the previous simulation.
P-SDR is excluded from the plot as its complexity is orders of magnitude higher than the ADMM-based schemes.
Across all scenarios, the projection-based schemes consistently exhibit higher initial complexity due to the preprocessing overhead of matrix projection, while the iteration-based schemes start with significantly lower complexity but grow rapidly with iterations.
Given that the proposed algorithms typically require few iterations to attain upper performance, the iteration-based schemes demonstrate a clear computational advantage, particularly in scenarios where sensing interference is not severe.


\begin{figure}[t]
	\centering
	\vspace{-0.3cm}
	\includegraphics[width=0.48\textwidth]{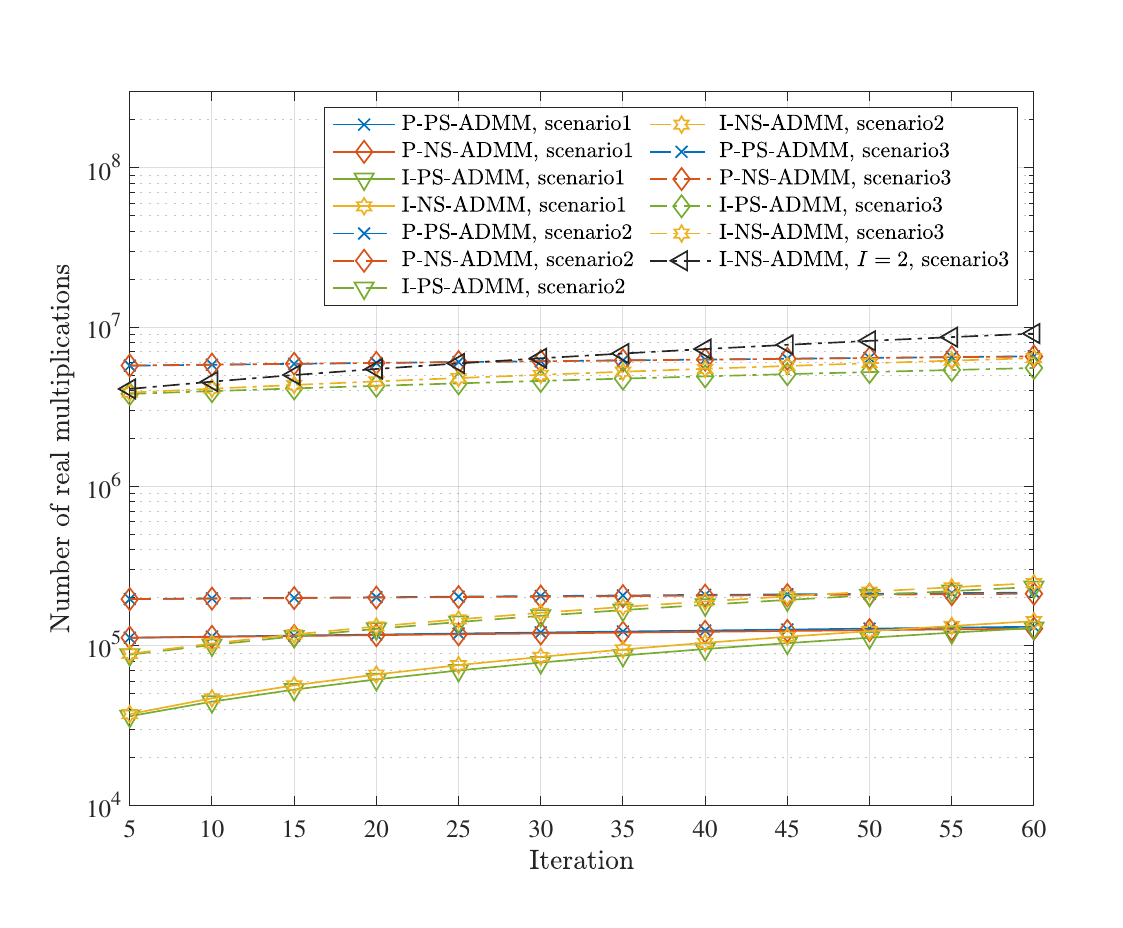}
	\vspace{-0.5cm}
	\caption{Comparisons of complexity in scenario $1$ ($32 \times 8 \times 2$), scenario $2$ ($32 \times 8 \times 8$), and scenario $3$ ($64 \times 60 \times 2$).}
	\vspace{-0.3cm}
	\label{complexity_comparison}
\end{figure}

Fig.~\ref{ml} compares the performance of the NS-aided schemes and P-SDR with ML detection (implemented via classic sphere decoding \cite{damen2003maximum}) across 4-QAM and 16-QAM modulation in a $128 \times 16 \times 2$ MIMO-ISAC system.
It can be observed that P-NS-ADMM achieves near-ML detection performance under both 4-QAM and 16-QAM modulation schemes, indicating that it exploits full received diversity, consistent with the result derived in Theorem~\ref{full_diversity}.
Furthermore, in the 16-QAM case, a noticeable performance gap exists between I-NS-ADMM and ML detection, confirming that I-NS-ADMM fails to achieve full received diversity order.

\begin{figure}[t]
	\vspace{-0.3cm}
	\centering
	\includegraphics[width=0.48\textwidth]{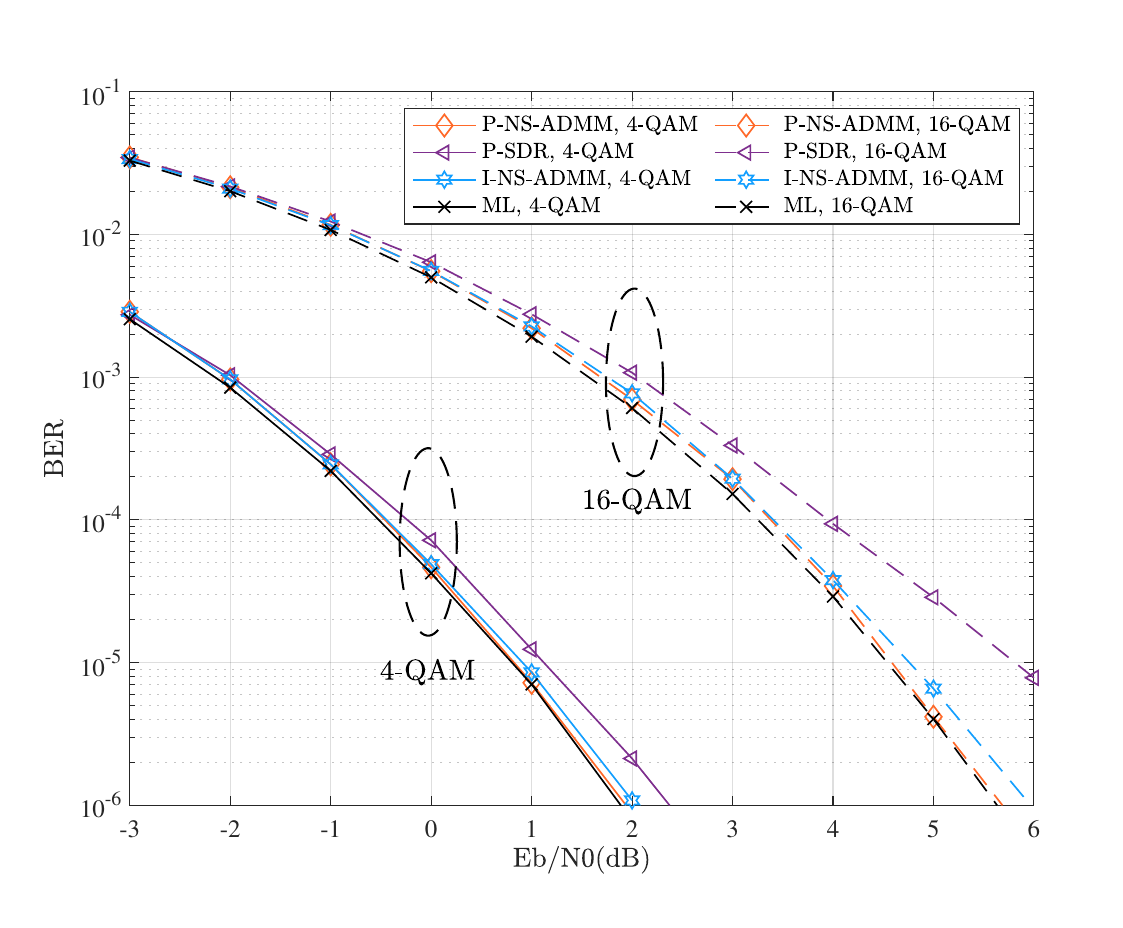}
	\vspace{-0.5cm}
	\caption{Comparisons of BER performance for 16-QAM and 4-QAM in $128 \times 16 \times 2$ MIMO-ISAC system.}
	\vspace{-0.3cm}
	\label{ml}
\end{figure}

Fig.~\ref{imperfect} compares the BER performance under an imperfect angle prior with a $1^{\circ}$ error.
As observed, both P-SDR and the proposed NS-aided schemes experience performance degradation due to the residual sensing interference caused by the mismatched projection or steering matrices.
Nevertheless, the proposed P-NS-ADMM and I-NS-ADMM schemes maintain a consistent performance advantage over P-SDR.
This indicates that the neighborhood search mechanism preserves its effectiveness in the presence of angle prior errors.

\begin{figure}[t]
\vspace{-0.3cm}
	\centering
	\includegraphics[width=0.48\textwidth]{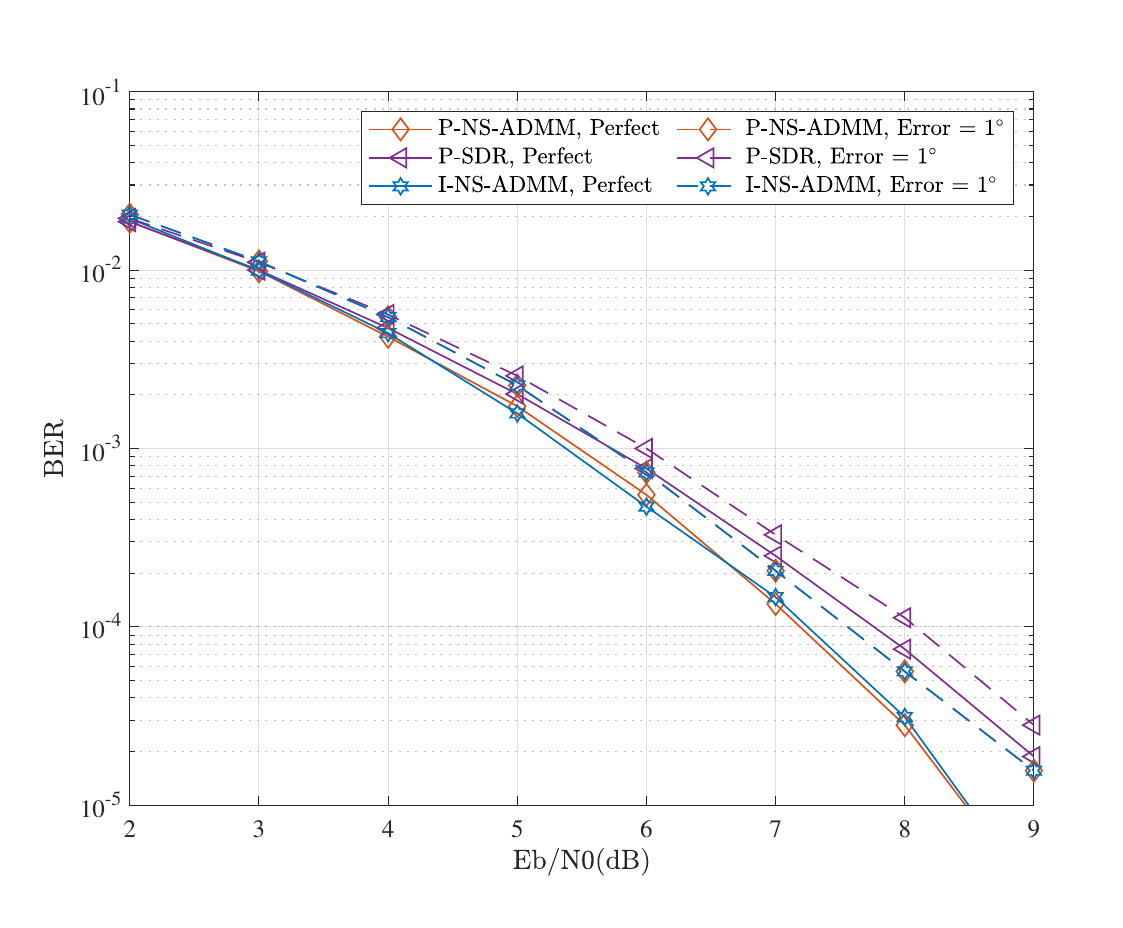}
	\vspace{-0.5cm}
	\caption{Comparisons of BER performance with perfect and imperfect angle prior information (error = $1^{\circ}$) in $32 \times 8 \times 2$ MIMO-ISAC system.}
	\vspace{-0.3cm}
	\label{imperfect}
\end{figure}

Fig.~\ref{sen_128_32} evaluates the sensing performance across 16-QAM modulation schemes in $32 \times 8 \times 2$ and $32 \times 8 \times 8$ MIMO-ISAC systems.
The sensing-only scheme without communication interference is used as a comparison.
It is evident that the estimation of the reflection coefficient is influenced by the performance of communication detection.
In this case, the proposed NS-aided schemes exhibit superior sensing performance compared with P-SDR and other ADMM-based schemes.
Meanwhile, the NMSE degrades with more targets (from $M_t=2$ to $M_t=8$) due to the reduced power allocation per target under a fixed total power constraint.

\begin{figure}[t]
	\vspace{-0.3cm}
	\centering
	\includegraphics[width=0.48\textwidth]{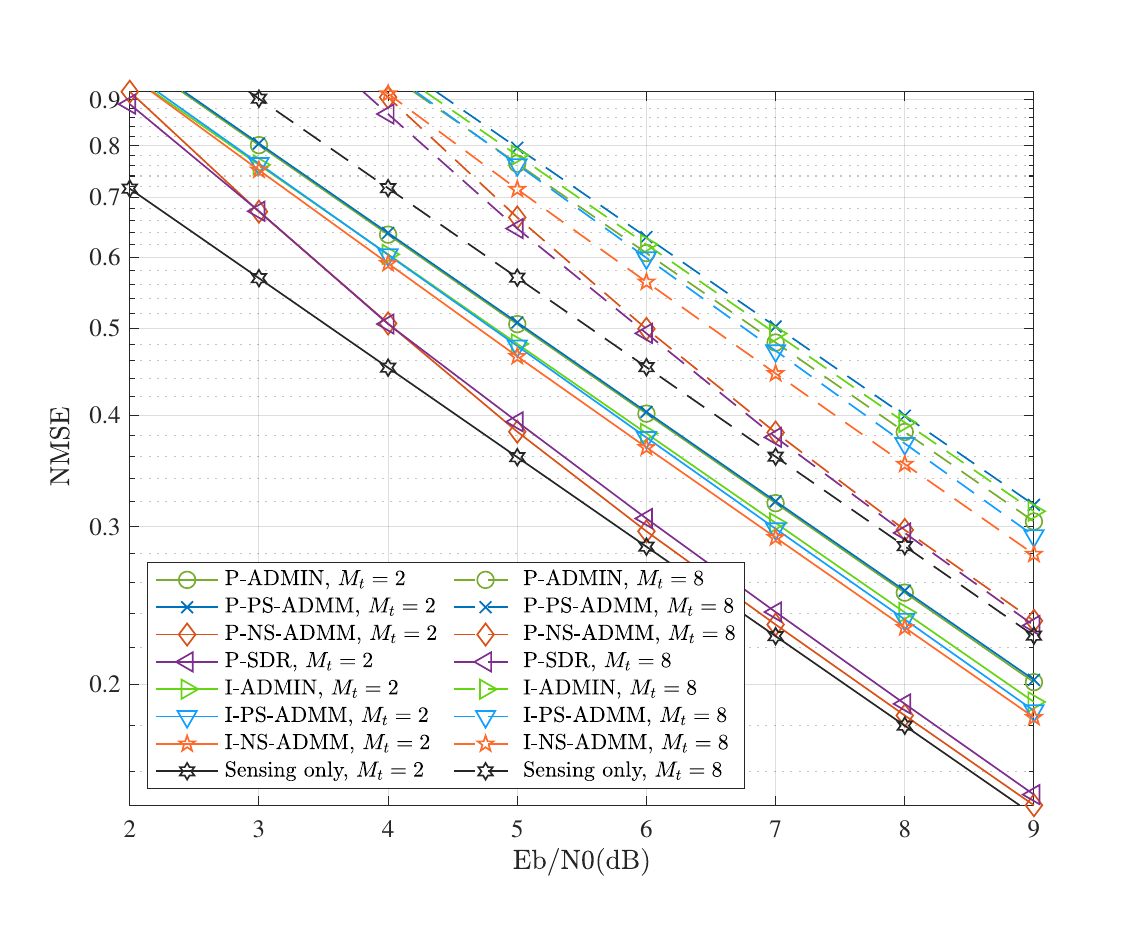}
	\vspace{-0.5cm}
	\caption{Comparisons of NMSE performance for 16-QAM in $32 \times 8 \times 2$ and $32 \times 8 \times 8$ MIMO-ISAC system.}
	\label{sen_128_32}
	\vspace{-0.3cm}
\end{figure}

Fig.~\ref{sen_64_60_2} extends the NMSE performance comparison to the $64 \times 60 \times 2$ MIMO-ISAC system.
It is clear that P-NS-ADMM and I-NS-ADMM demonstrate significantly different performance across low and high SNR regimes.
For P-NS-ADMM, substantial performance degradation is observed under low SNR conditions.
This is because P-NS-ADMM performs communication detection followed by the estimation of the real-valued sensing parameters.
However, NS projects the real values to discrete integers, which increases the NMSE due to the high BER at low SNR.
Since the BER decreases with increasing SNR, the sensing performance of P-NS-ADMM improves and approaches the sensing-only scheme. 
In comparison, within the I-NS-ADMM framework, the sensing parameters are output in each iteration before the execution of NS, thereby ensuring accurate estimation of the sensing parameters at low SNR.
Nevertheless, I-NS-ADMM is constrained by the mutual interference between S\&C signals, resulting in a discrepancy with the sensing-only scheme.

\begin{figure}[t]
	\vspace{-0.3cm}
	\centering
	\includegraphics[width=0.48\textwidth]{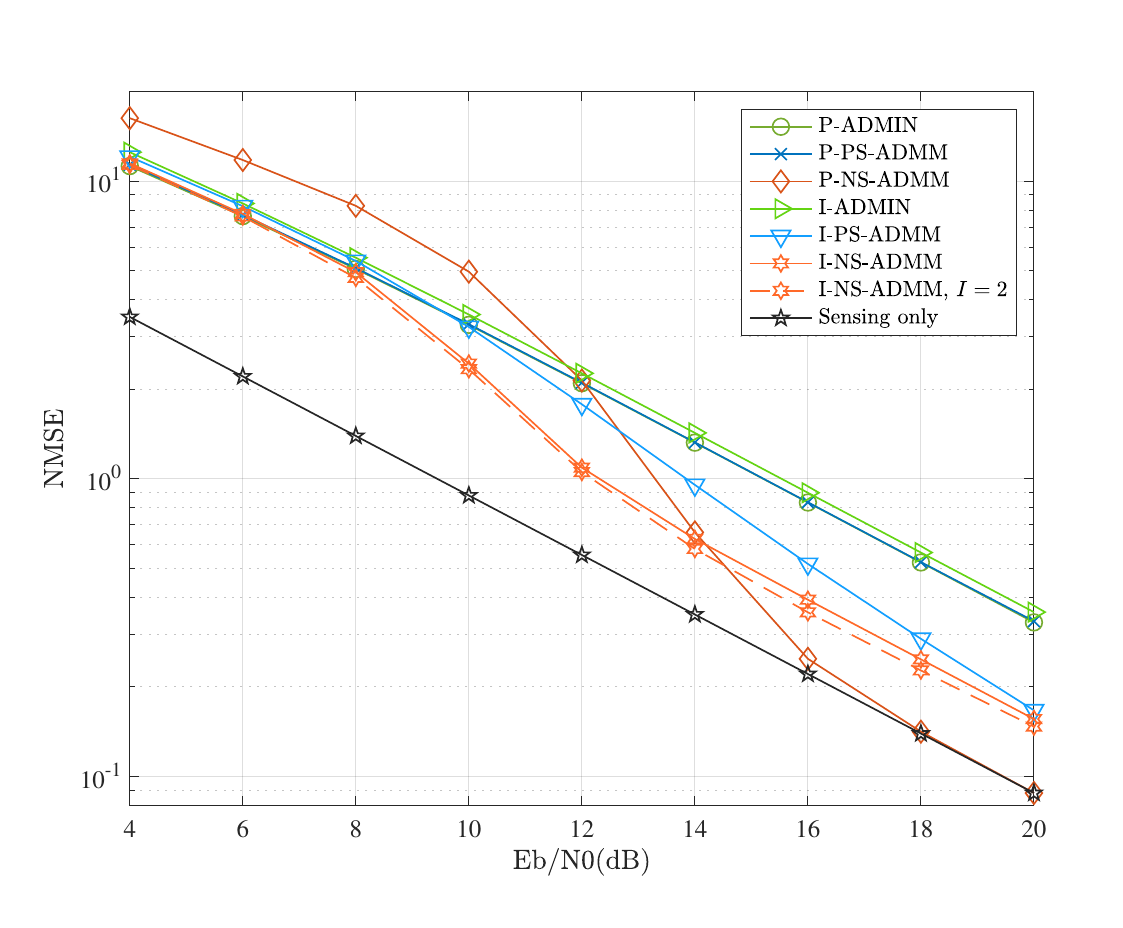}
	\vspace{-0.5cm}
	\caption{Comparisons of NMSE performance for 16-QAM in $64 \times 60 \times 2$ MIMO-ISAC system.}
	\label{sen_64_60_2}
	\vspace{-0.3cm}
\end{figure}

Given that we have obtained the expression of $\mathbf{x}_{\mathrm{NS}}$ in \eqref{xns_new}, subtracting $\mathbf{x}_{\mathrm{NS}}$ from \eqref{original}, we can obtain
\begin{align}
	\mathbf{y} - \mathbf{Hx}_{\mathrm{NS}} &= \mathbf{As} + \mathbf{H}(\mathbf{x} - \mathbf{x}_{\mathrm{NS}}) + \mathbf{n} \nonumber \\
	&= \mathbf{As} -\mathbf{H}(\tilde{\mathbf{D}}^{-1}\odot{\tilde{\mathbf{H}}^T\tilde{\mathbf{n}}}) + \mathbf{n}.
\end{align}
Therefore, the estimated sensing signal vector $\hat{\mathbf{s}}$ can be expressed as
\begin{align}
	\hat{\mathbf{s}} &= (\mathbf{A}^T \! \mathbf{A})^{-1}\mathbf{A}^T(\mathbf{y} \!-\! \mathbf{Hx}_{\mathrm{NS}}) \nonumber \\
	& = \mathbf{s} \!+\! (\mathbf{A}^T \!\mathbf{A})^{-1}\mathbf{A}^T \!\mathbf{n} \!-\! (\mathbf{A}^T \! \mathbf{A})^{-1}\!\mathbf{A}^T\mathbf{H}(\tilde{\mathbf{D}}^{-1}\! \odot{\tilde{\mathbf{H}}^T \! \tilde{\mathbf{n}}}),
	\label{senreference}
\end{align}
where the third term represents the effect of communication detection on sensing estimation.
It can be eliminated only if the communication and sensing channels are orthogonal to each other, namely $\mathbf{A}^T\mathbf{H}=\mathbf{0}$. 
Since $\mathbf{x}_{\mathrm{NS}}$ in \eqref{xns_new} achieves full received diversity order, \eqref{senreference} can be considered as a performance benchmark at high SNR.

Fig.~\ref{sen_64_60_2_pcps_noise_20_30} compares the NMSE performance for 16-QAM in $64 \times 60 \times 2$ MIMO-ISAC system with $E_b/N_0 = 20~\mathrm{dB}$ and $E_b/N_0 = 30~\mathrm{dB}$.
It can be noticed that the NMSE decreases with increasing sensing power.
In addition, the NMSE performance of I-PS-ADMM and I-NS-ADMM approaches the benchmark, while P-NS-ADMM is superior to the benchmark, which reflects the performance advantages of these three algorithms.
Meanwhile, at high SNR, P-NS-ADMM successfully decodes the communication signals, which makes its performance almost the same as the sensing only scheme.
Consequently, in a given scenario where the performance of I-NS-ADMM is found to be inferior to the benchmark, despite the complexity advantage of I-NS-ADMM, one can consider P-NS-ADMM.

\begin{figure}[t]
	\vspace{-0.3cm}
	\centering
	\includegraphics[width=0.48\textwidth]{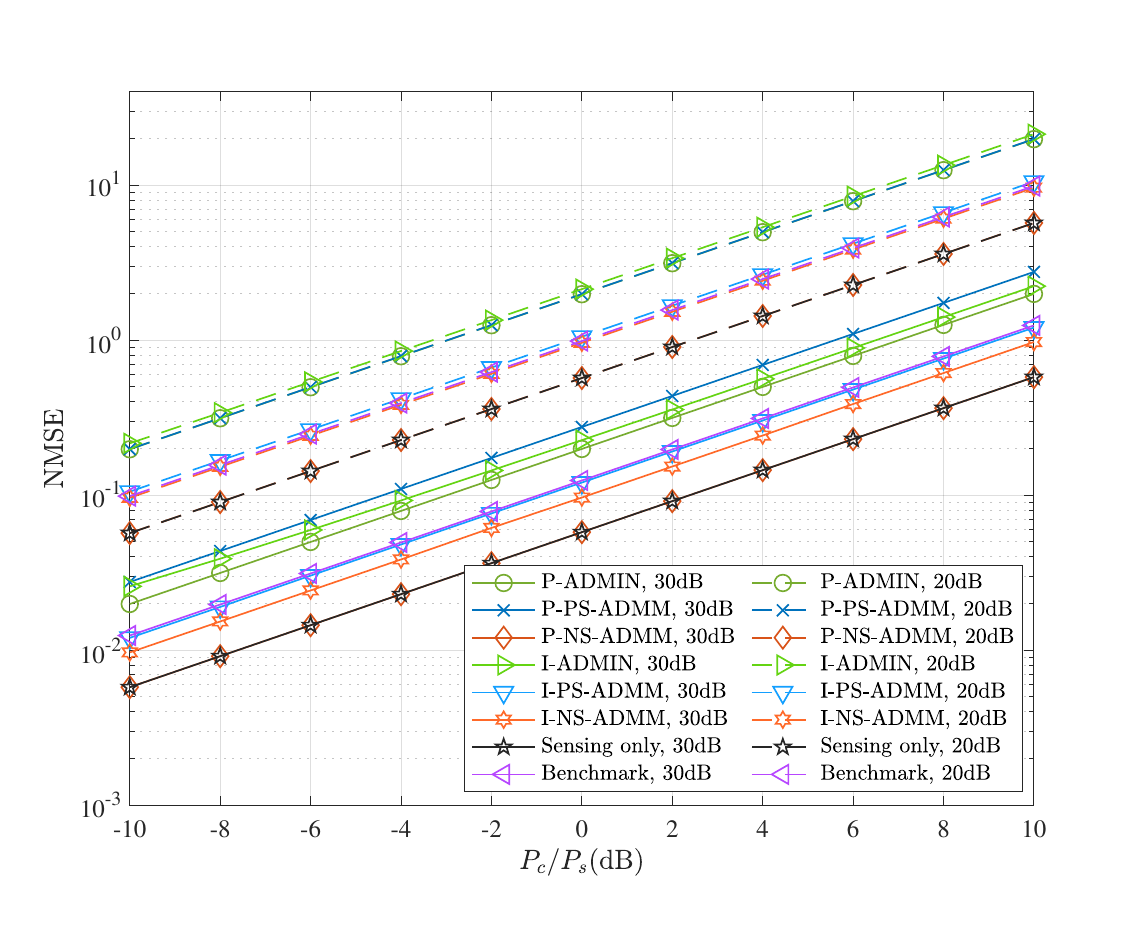}
	\vspace{-0.5cm}
	\caption{Comparisons of NMSE performance for 16-QAM in $64 \times 60 \times 2$ MIMO-ISAC system with $E_b/N_0 = 20~\mathrm{dB}$ and $E_b/N_0 = 30~\mathrm{dB}$.}
	\vspace{-0.3cm}
	\label{sen_64_60_2_pcps_noise_20_30}
\end{figure}

In summary, P-NS-ADMM and I-NS-ADMM each have distinct advantages and disadvantages, making them suitable for different scenarios:
\begin{itemize}
	\item The P-NS-ADMM algorithm demonstrates superior stability and performance in high SNR regimes as well as scenarios where $N_r$ approaches $M_t + U$.
	However, due to the projection operation, it has a higher complexity.
	\item The I-NS-ADMM algorithm offers reduced computational complexity and exhibits enhanced performance under low SNR conditions.
	Nevertheless, its effectiveness is limited in high SNR scenarios due to S\&C mutual interference.
\end{itemize}

\section{Conclusion}
In this paper, we modeled the detection problem as a MILS problem for large-scale MIMO-ISAC systems in multi-user and multi-target scenarios.
To mitigate mutual interference between S\&C signals, P-NS-ADMM was proposed, which utilizes the projection operation transforming the MILS problem into a traditional communication detection problem.
Then, we presented that the P-NS-ADMM algorithm achieves the same received diversity order as ML detection and matches the optimal performance with the increment of SNR.
In addition, to further reduce the complexity, the low-complexity I-NS-ADMM algorithm was proposed, which tackles the MILS problem directly.
Moreover, a flexible mechanism of ADMM iterations was also given for a better estimation of the sensing signals.
In simulations, we show that our proposal has remarkable advantages over other ADMM-based schemes in terms of BER and NMSE.
Meanwhile, we summarized the advantages and disadvantages of both schemes under various scenarios.
Therefore, for large-scale MIMO-ISAC systems, the proposed NS-aided ADMM detection schemes provide a realistic alternative that achieves significant performance.

\bibliographystyle{IEEEtran}
\bibliography{IEEEabrv,reference1}

\end{document}